\documentclass[cmp]{svjour}
\usepackage{graphics}
\usepackage{color}
\usepackage[usenames,dvipsnames]{xcolor}
\usepackage{amsxtra}
\usepackage{upref}
\usepackage{amsfonts}
\usepackage{mathrsfs}
\usepackage{euscript}
\usepackage{array}
\usepackage{stmaryrd}
\usepackage{hyperref}

\newcounter{Step}
\newenvironment{step}[0]{\bigskip\addtocounter{Step}{1}\noindent\textbf{Step \theStep :} }{\
  \begin{flushright} \end{flushright}}

\newcommand{\Macto}{{\overset{M}{\rightarrow}\,}}

\newcommand{\frakF}{{\mathfrak{F}}}
\newcommand{\dotT}{\dot{T}}

\newcommand{\telque}{{\,;\,}}

\newcommand{\bbC}{{\mathbb{C}}}
\newcommand{\bbK}{{\mathbb{K}}}
\newcommand{\bbR}{{\mathbb{R}}}
\newcommand{\bbN}{{\mathbb{N}}}

\newcommand{\calB}{{\mathcal{B}}}

\newcommand{\calD}{{\mathcal{D}}}
\newcommand{\calE}{{\mathcal{E}}}

\newcommand{\calT}{{\mathcal{T}}}

\newcommand{\supp}{{\mathrm{supp\,}}}

\newcommand{\WF}{{\mathrm{WF}}}

\newcommand{\calS}{{\mathcal{S}}}

\newtheorem{prop}{Proposition}
\newtheorem{thm}[prop]{Theorem}
\newtheorem{cor}[prop]{Corollary}

\newtheorem{lem}[prop]{Lemma}
\newtheorem{dfn}[prop]{Definition}

\begin{document}
\title{Functional properties of H{\"o}rmander's space of 
  distributions having a specified wavefront set}
\titlerunning{Functional properties of $\calD'_\Gamma$}
\author{Yoann Dabrowski\inst{1} \and Christian Brouder\inst{2}
}                     
\institute{
Institut Camille Jordan UMR~5208,
Universit\'{e} de Lyon, Universit\'{e} Lyon 1,\\
43 bd. du 11 novembre 1918, F-69622 Villeurbanne cedex,
France \and 
Institut de Min\'eralogie, de Physique des Mat\'eriaux et de
Cosmochimie, Sorbonne Universit\'es,
UMR CNRS 7590, UPMC Univ. Paris 06, Mus\'eum National d'Histoire
Naturelle, IRD UMR 206,
4 place Jussieu, F-75005 Paris, France.
}

\date{Received: date / Accepted: date}
%
\communicated{M. Salmhofer}
\maketitle
\begin{abstract}
The space  $\calD'_\Gamma$ of distributions 
having their wavefront
sets in a closed cone $\Gamma$ has become important in physics
because of its role in the formulation of quantum field theory
in curved spacetime. In this paper, the topological and bornological 
properties of $\calD'_\Gamma$ and its dual $\calE'_\Lambda$ are 
investigated. It is found that $\calD'_\Gamma$ is a 
nuclear, semi-reflexive and semi-Montel complete normal space
of distributions.
Its strong dual $\calE'_\Lambda$ is a 
nuclear, barrelled and (ultra)bornological normal space of distributions
which, however, is not even sequentially complete.
Concrete rules are given to determine whether a distribution
belongs to $\calD'_\Gamma$, whether a sequence converges in
$\calD'_\Gamma$ and whether a set of distributions is bounded
in $\calD'_\Gamma$.
\end{abstract}

\section{Introduction}
\label{mul-dis-sect}

Standard quantum field theory uses
 Feynman diagrams
in the momentum space. However, this framework 
is not suitable for quantum field theory in arbitrary
spacetimes because of
the absence of translation invariance.
In 1992, Radzikowski \cite{Radzikowski-92-PhD,Radzikowski} showed 
the wavefront set of distributions 
to be  a key
concept to describe quantum fields in curved spacetime.
This idea was developed into a rigorous renormalized
scalar field theory in curved spacetime by
Brunetti and Fredenhagen~\cite{Brunetti2}, followed by
Hollands and Wald~\cite{Hollands2}.
This approach was rapidly extended to 
deal with
Dirac fields~\cite{Kratzert-00,Hollands-01,Antoni-06,%
Dappiaggi-09,Sanders-10-Dirac,Rejzner-11},
gauge fields~\cite{Hollands-08,Fredenhagen-11,Fredenhagen-13}
and even the quantization of gravitation~\cite{Brunetti-13-QG}.

This tremendous progress was made possible by
a complete reformulation of quantum field theory,
where the wavefront set of distributions plays a
central role, for example to determine the algebra
of microcausal functionals, to define a
spectral condition for time-ordered products
and quantum states and to give a rigorous
description of renormalization.

In other words, the natural space where quantum field
theory takes place is not the space of distributions
$\calD'$, but the space $\calD'_\Gamma$ of distributions having
their wavefront set in a 
specified closed cone $\Gamma$.
This space and its simplest properties were described
by H{\"o}rmander in 1971~\cite{Hormander-71}.
Since $\calD'_\Gamma$ is now a crucial tool
of quantum field theory, it is important
to investigate its topological and functional
properties. 
For example, renormalized time-ordered
products are determined as an extension of a distribution
to the thin diagonal. Since this extension is defined as the
limit of a sequence, we need simple criteria to determine
the convergence of a sequence in $\calD'_\Gamma$.
The ambiguity of renormalization is determined, among other things,
by the way this distribution varies under scaling. 
Scaled distributions are defined with respect to a bounded set
in $\calD'_\Gamma$.  
Thus, we need simple tests to know when a set
of distributions is bounded.
The purpose of this paper is to provide tools
to answer these questions in a simple way.

The wavefront set of distributions plays also a key role
in microlocal analysis,
to determine whether a distribution can be pulled back,
restricted to a submanifold or multiplied by another
distribution~\cite[Chapter 8]{HormanderI}. Therefore,
the wavefront set has become a standard subject
in textbooks of distribution theory and microlocal
analysis~\cite{HormanderI,Duistermaat,Guillemin,%
  Chazarain,ReedSimonII,%
Friedlander,Strichartz-03,Grigis,Strohmaier,Eskin,Wagschal-11}.
However, to the best of our knowledge, no detailed study was
published on the functional properties of $\calD'_\Gamma$.

Many properties of $\calD'_\Gamma$ will be deduced
from properties of its dual. Thus, we shall first
calculate the dual of $\calD'_\Gamma$,
denoted by $\calE'_\Lambda$, which turns out to be
the space of compactly supported distributions having
their wavefront set included in an \emph{open cone}
$\Lambda$ which is the complement of $\Gamma$ up to
a change of sign. 
Such a space $\calE'_\Lambda$ is used
in quantum field theory to define microcausal
functionals~\cite{Fredenhagen-11}.

We now summarize our main results. Although they are both
nuclear and normal spaces of distributions,
$\calD'_\Gamma$ and $\calE'_\Lambda$ have very contrasted
properties; (i) $\calD'_\Gamma$ is semi-reflexive and complete while
$\calE'_\Lambda$ is not even sequentially complete;
(ii) $\calE'_\Lambda$ is barrelled,
and ultrabornological, while
$\calD'_\Gamma$ is neither barrelled nor bornological.
For applications, the most significant property of $\calD'_\Gamma$
is to be semi-Montel. 
Indeed, two steps involving $\calD'_\Gamma$ are particularly 
important in the renormalization
process described by Brunetti and Fredenhagen~\cite{Brunetti2}.
The first step is a control
of the divergence of the relevant distributions
near the diagonal: there must be a real number $s$ such that
the family  $\{\lambda^{-s} u_\lambda\}_{0<\lambda \le 1}$ 
is a bounded set of distributions,
where $u_\lambda$ is a scaled distribution. 
This proof is facilitated by our determination of bounded sets:
\begin{prop}
A set $B$ of distributions in $\calD'_\Gamma$ is bounded
if and only if, for every $v\in \calE'_\Lambda$, there is
a constant $C_v$ such that $|\langle u,v\rangle| < C_v$
for all $u\in B$. Such a weakly bounded set is also
strongly bounded and equicontinuous.
Moreover, the closed bounded sets
of $\calD'_\Gamma$ are compact, complete and metrizable.
\end{prop}
The second step is the proof that the extension of a distribution
can be defined as the limit of a sequence of distributions
in $\calD'_\Gamma$. For this we derive the following 
convergence test:
\begin{prop}
If $u_i$ is a sequence of elements of $\calD'_\Gamma$ such
that, for any $v\in \calE'_\Lambda$, the 
sequence  $\langle u_i,v\rangle$ converges in $\bbC$ to 
a number $\lambda_v$, then $u_i$ converges to a distribution
$u$ in $\calD'_\Gamma$ and $\langle u,v\rangle=\lambda_v$
for all $v\in \calE'_\Lambda$.
\end{prop}

We now describe the organization of the paper.
After this introduction, we determine a pairing between
$\calD'_\Gamma$ and $\calE'_\Lambda$ and we show that
this pairing is compatible with duality.
Then, we prove that $\calD'_\Gamma$ is a normal
space of distributions.
The next section investigates several topologies 
on $\calE'_\Lambda$ and shows their equivalence.
Then, the nuclear and bornological properties
of $\calD'_\Gamma$ and $\calE'_\Lambda$ are discussed.
Bornology enables us to prove that $\calD'_\Gamma$ is complete
and it is relevant to the problem of quantum field theory on curved
spacetime because some isomorphisms of the space of sections of a vector
bundle over a manifold are stronger in the bornological setting than
in the topological one (see section~\ref{borno-sect}).
These results are put together to determine the
main functional properties of $\calD'_\Gamma$ and its dual.
Finally, a counter-example is constructed to 
show that $\calE'_\Lambda$ is not sequentially complete.
This will imply that $\calD'_\Gamma$ and its dual do not
enjoy all the nice properties of $\calD'$.

\section{The dual of $\calD'_\Gamma$}
In this section, we review what is known about
the topology of $\calD'_\Gamma$
and we describe the functional analytic tools
(duality pairing and normal spaces of distributions)
that enable us to investigate the dual
of $\calD'_\Gamma$.

\subsection{What is known about $\calD'_\Gamma$}
\label{wwk-sect}

Let us fix the notation.
Let $\Omega$ be an open set in $\bbR^n$, 
we denote by $T^*\Omega$ the 
cotangent bundle over $\Omega$,
by $UT^*\Omega=\{(x;k)\in T^*\Omega\telque |k|=1\}$
(where $|k|$ is the standard Euclidian norm on $\bbR^n$)
the sphere bundle over $\Omega$ and by
$\dotT^*\Omega=T^*\Omega\backslash\{(x;0)\telque x\in \Omega\}$
the cotangent bundle without the zero section.
We say that a subset $\Gamma$ of $\dotT^*\Omega$ is
a cone if $(x;\lambda k)\in \Gamma$ whenever
$(x;k)\in \Gamma$ and $\lambda>0$ and such a cone
is said to be closed if it is closed in $\dotT^*\Omega$.
For any closed cone $\Gamma$,
H\"ormander defined~\cite[p.~125]{Hormander-71}
the space $\calD'_\Gamma$ to be the set of 
distributions in $\calD'(\Omega)$ 
having their wavefront set in $\Gamma$.
He also described what he called a \emph{pseudo-topology}
on $\calD'_\Gamma$, which means that he defined a concept of 
convergence in $\calD'_\Gamma$ but not a topology (as a family of open sets).
His definition was equivalent to the following 
one~\cite[p.~262]{HormanderI}: a sequence
$u_j\in \calD'_\Gamma$ converges to $u\in \calD'_\Gamma$
if 
\begin{itemize}
\item[(i)] The sequence of numbers $\langle u_j,f\rangle$ 
  converges to $\langle u,f\rangle$ in the ground field
  $\bbK$ (i.e. $\bbR$ or $\bbC$) for all $f\in \calD(\Omega)$.
\item[(ii)] If $V$ is a closed cone in $\bbR^n$
  and $\chi$ is an element of $\calD(\Omega)$ that satisfy
  $(\supp\chi\times V)\cap \Gamma=\emptyset$, then
  $\sup_{k\in V}(1+|k|)^N
  |\widehat{u_j\chi}(k)-\widehat{u\chi}(k)|\to 0$
  for all integers $N$,
\end{itemize}
where a hat over a distribution 
(e.g. $\widehat{u\chi}$) denotes its
Fourier transform (the
Fourier transform of $f\in \calD(\Omega)$ being
defined by
$\hat{f}(k) = \int_\Omega  e^{i k\cdot x} f(x) dx$).

H\"ormander then showed that $\calD(\Omega)$ is dense
in $\calD'_\Gamma$. More precisely, for every $u\in \calD'_\Gamma$
there is a sequence of functions $u_j\in \calD(\Omega)$
such that $u_j$ converges to $u$ in the 
above sense~\cite[p.~262]{HormanderI}.
This concept of convergence is compatible with
different topologies. 
The topology of $\calD'_\Gamma$ used in the
literature~\cite{Grigis,Cardoso-75,Duistermaat,%
 Grudzinski,Guillemin,Alesker-10}, which is
usually called the H\"ormander 
topology~\cite{Brunetti2,Strohmaier,Stottmeister-14}),
is that of a locally convex topological vector space defined
by the following seminorms:
\begin{itemize}
\item[(i)] $p_f(u)=|\langle u,f\rangle|$
  for all $f\in \calD(\Omega)$.
\item[(ii)] 
  $||u||_{N,V,\chi}=\sup_{k\in V}(1+|k|)^N
  |\widehat{u\chi}(k)|$,
  for all integers $N$, all closed cones $V$
  and all $\chi\in \calD(\Omega)$ such that 
  $(\supp\chi\times V)\cap \Gamma=\emptyset$.
\end{itemize}
We immediately observe that $\calD'_\Gamma$
is a Hausdorff locally convex space 
because $u=0$ if $p_i(u)=0$ for all its seminorms
$p_i$~\cite[p.~96]{Horvath}. Indeed, if
$p_f(u)=|\langle u,f\rangle|=0$ for all $f\in \calD(\Omega)$,
then $u=0$.  
When we speak of ``all the seminorms'' of a 
locally convex space $E$,
we mean all the seminorms of a family of seminorms 
defining the topology of $E$~\cite[p.~63]{Treves}.

\subsection{Duality pairing}

Mackey's duality theory~\cite{Mackey-43,Mackey-43-2,%
Mackey-45,Mackey-46} is a powerful technique to investigate
the topological properties of locally convex
spaces~\cite{Bourbaki-TVS,Horvath}. The first step of this method is to find
a duality pairing between two spaces.

Let us take the example of the duality pairing between
$\calD'(\Omega)$ and $\calD(\Omega)$. Any test
function $u\in\calD(\Omega)$ can be paired to any $f\in \calD(\Omega)$ by
$\langle u,f\rangle= \int_\Omega u(x)f(x) dx$.
The density of $\calD(\Omega)$ in $\calD'(\Omega)$ implies that
this pairing can be uniquely extended to a pairing
between $\calD'(\Omega)$ and $\calD(\Omega)$, also denoted by
$\langle u,f\rangle$, that can be written
\begin{eqnarray}
\langle u,f\rangle &=& \frac{1}{(2\pi)^n} 
\int_{\bbR^n} \widehat{u\varphi}(k)\hat{f}(-k) dk,
\label{uf}
\end{eqnarray}
where the function $\varphi\in\calD(\Omega)$ is equal
to 1 on a compact neighborhood of the support of $f$.
Indeed, $\langle u,f\rangle=
\langle \varphi u,f\rangle$~\cite[p.~90]{Schwartz-66} and
$\varphi u$ has a Fourier transform because it 
is a compactly supported distribution~\cite[p.~165]{HormanderI}.
This pairing is compatible with duality, in the sense
that any element $\alpha$ in the topological dual of
$\calD(\Omega)$ can be written $\alpha(f)=\langle u,f\rangle$
for one element $u$ of $\calD'(\Omega)$, by definition
of the space of distributions.

We would like to find a similar pairing between $\calD'_\Gamma$
and another space to be determined.
Grigis and Sj{\"o}strand~\cite[p.~80]{Grigis} showed
that the pairing 
$\langle u,v\rangle=\int_\Omega u(x) v(x) dx$
between $C^\infty(\Omega)$ and $\calD(\Omega)$
extends uniquely to the pairing defined by
eq.~(\ref{uf}) between $\calD'_\Gamma$ and 
every space  $\calE'_\Xi$ of compactly supported
distributions whose wavefront set is contained in $\Xi$,
where $\Xi$ is any closed cone such that
$\Gamma'\cap\Xi=\emptyset$,
where $\Gamma'=\{(x;k)\in \dotT^*\Omega\telque
   (x;-k)\in \Gamma\}$
(see also \cite[p.~512]{Chazarain} for a similar result).

We need to slightly extend their definition by pairing
$\calD'_\Gamma$ with the space $\calE'_\Lambda$, where $\Lambda$ is
now the \emph{open} cone $\Lambda=(\Gamma')^c$. Note that this
space is the union of the ones considered by Grigis and
Sj{\"o}strand.
The next lemma does not contain more information than their result, but, for the reader's convenience, we first show that this extended pairing is well defined.
\begin{lem}
\label{pairinglem}
If $\Gamma$ is a closed cone in $\dotT^*\Omega$ and
$\Lambda=(\Gamma')^c=\{(x;k)\in \dotT^*\Omega\telque
(x,-k)\notin \Gamma\}$, then
the following pairing between $\calD'_\Gamma$ and 
$\calE'_\Lambda=\{v\in \calE'(\Omega)\telque \WF(v)\subset \Lambda\}$
is well defined:
\begin{eqnarray*}
\langle u,v\rangle &=& \frac{1}{(2\pi)^n} 
\int_{\bbR^n} \widehat{u\varphi}(k)\hat{v}(-k) dk,
\end{eqnarray*}
where $u\in \calD'_\Gamma$, $v\in \calE'_\Lambda$ and
$\varphi$ is any function in $\calD(\Omega)$ equal
to 1 on a compact neighborhood of the support of $v$.
This pairing is separating and, for any $v\in \calE'_\Lambda$,
the map $\lambda:\calD'_\Gamma\to\bbK$
defined by $\lambda(u)=\langle u,v\rangle$ is continuous.
\end{lem}
\begin{proof}
We first consider the case where $\Gamma$ is neither
empty nor $\dotT^*\Omega$.
A distribution $v\in\calE'_\Lambda$ is compactly supported
and its wavefront set is a closed cone contained in $\Lambda$,
which implies $\WF(v)\cap\Gamma'=\emptyset$.
The product of distributions $uv$ is then a well-defined
distribution by H\"ormander's theorem~\cite[p.~267]{HormanderI}.
We estimate now
$\langle u,v\rangle=(2\pi)^{-n}\int \widehat{u\varphi}(k)\hat{v}(-k) dk$.


By a classical construction~\cite[p.~61]{Eskin}, there is a finite set
of non-negative smooth functions $\psi_j$ such that
$\sum_j \psi_j^2=1$ on a compact neighborhood $K$ of the support of $v$
and there are closed cones  $V_{uj}$ and $V_{vj}$
that satisfy the three conditions: (i) $V_{uj}\cap (-V_{vj})=\emptyset$,
(ii) $\supp\psi_j\times V_{uj}^c\cap \Gamma=\emptyset$ and
(iii) $\supp\psi_j\times V_{vj}^c\cap \WF(v)=\emptyset$.
As a consequence of these conditions, we have
$\Gamma|_K \subset \cup_j\,(\supp\psi_j\times V_{uj})$ and
$\WF(v) \subset \cup_j \,(\supp\psi_j\times V_{vj})$.
If we choose $\varphi=\sum_j\psi_j^2$ we can write
$\langle u,v\rangle=\sum_j I_j$, where
$I_j=(2\pi)^{-n}\int \widehat{u\psi_j}(k)\widehat{v\psi_j}(-k) dk$.

Following again Eskin~\cite[p.~62]{Eskin}, we can define 
homogeneous functions
of degree zero $\alpha_j$ and $\beta_j$ on 
$\bbR^n$,
which are smooth except at the origin, measurable,
non-negative and bounded by 1 on $\bbR^n$ and such that
$\supp\alpha_j$ and $\supp\beta_j$ are closed cones
satisfying the three conditions (i), (ii)
and (iii) stated above, with
$\alpha_j=1$ on $V_{uj}$ and $\beta_j=1$ on $V_{vj}$.
Then we insert $1=\big(\alpha_j+(1-\alpha_j)\big)
\big(\beta_j+(1-\beta_j)\big)$ in the integral defining
$I_j$ and we obtain 
$I_j=I_{1j}+I_{2j}+ I_{3j}+I_{4j}$, where
\begin{eqnarray*}
I_{1j} &=& (2\pi)^{-n} \int_{\bbR^n} 
  \alpha_j(-k) \widehat{\psi_j u}(-k)\,
  \beta_j(k) \widehat{\psi_j v}(k) \,dk,
\\
I_{2j} &=& (2\pi)^{-n} \int_{\bbR^n} 
  \alpha_j(-k) \widehat{\psi_j u}(-k)\,
  \big(1-\beta_j(k)\big) \widehat{\psi_j v}(k) \,dk,
\\
I_{3j} &=& (2\pi)^{-n} \int_{\bbR^n} 
  \big(1-\alpha_j(-k)\big) \widehat{\psi_j u}(-k)\,
  \beta_j(k) \widehat{\psi_j v}(k) \,dk,
\\
I_{4j} &=& (2\pi)^{-n} \int_{\bbR^n} 
  \big(1-\alpha_j(-k)\big) \widehat{\psi_j u}(-k)\,
  \big(1-\beta_j(k)\big) \widehat{\psi_j v}(k) \,dk.
\end{eqnarray*}
We first notice that $I_{1j}=0$
because $(-\supp \alpha_j)\cap \supp\beta_j=\emptyset$.
We estimate $I_{4j}$.
The function $\beta_j$ was built so that
$(1-\beta_j)=0$ on $V_{vj}$
and $\supp\psi_j\times \supp(1-\beta_j)\cap\WF(v)=\emptyset$.
Then, for any integer $N$,
\begin{eqnarray*}
\big|\big(1-\beta_j(k)\big) \widehat{\psi_j v}(k)\big|
 & \le & ||v||_{N,U_{\beta j},\psi_j} (1+|k|)^{-N},
\end{eqnarray*}
where $U_{\beta j}=\supp(1-\beta_j)$.
Similarly
\begin{eqnarray}
\big|\big(1-\alpha_j(k)\big) \widehat{\psi_j u}(k)\big|
 & \le & ||u||_{M,U_{\alpha j},\psi_j} (1+|k|)^{-M},
\label{alphau}
\end{eqnarray}
where $U_{\alpha j}=\supp(1-\alpha_j)$.
Thus, for $N+M>n$,
\begin{eqnarray*}
|I_{4j}| & \le & ||u||_{M,U_{\alpha j},\psi_j} 
  ||v||_{N,U_{\beta j},\psi_j} 
  I_n^{N+M},
\end{eqnarray*}
where $I_n^N=(2\pi)^{-n} \int_{\bbR^n} (1+|k|)^{-N} dk$.

For $I_{3j}$ we use the fact that, $\psi_j v$ being
a compactly supported distribution, there is an integer
$m$ and a constant $C$ such that
$|\widehat{\psi_j v}(k)|\le C (1+|k|)^{m}$~\cite[p.~181]{HormanderI}.
When this estimate is combined with eq.~(\ref{alphau}) we obtain
for $M>n+m$,
\begin{eqnarray*}
|I_{3j}| & \le & ||u||_{M,U_{\alpha j},\psi_j} C
I_n^{M-m}.
\end{eqnarray*}
For the integral $I_{2j}$ we proceed differently
because we want to recover a seminorm of $\calD'_\Gamma$.
If we define $\hat{f}_j(k)=\alpha_j(-k)(1-\beta_j(k))
\widehat{\psi_j v}(k)$, then
$$I_{2j}=(2\pi)^{-n} \int \widehat{\psi_ju}(-k) \hat{f}_j(k) dk.$$
We call \emph{fast decreasing} a function $f(k)$ such that,
for every integer $N$, $|f(k)|\le C_N (1+|k|)^{-N}$ for 
some constant $C_N$. Note that our fast decreasing functions
are different from Schwartz rapidly decreasing functions.
The function $\hat{f}_j(k)$ is fast decreasing because
$\alpha_j$ and $\beta_j$ are bounded by 1,
$\widehat{\psi_j v}(k)$ is fast decreasing outside the
wavefront set of $v$ and $(1-\beta_j(k))$ cancels 
$\widehat{\psi_j v}(k)$ on this wavefront set.
The function $\hat{f}_j$ is also measurable because it is
the product of measurable functions.
Thus, by a standard result in the spirit
of~\cite[p.~145]{Friedlander},
its inverse Fourier transform $f_j$ 
exists and is smooth.
We can now rewrite $I_{2j}=\langle \psi_j u, f_j\rangle=
\langle u,\psi_jf_j\rangle$, which is well defined because
$\psi_j f_j$ is smooth and compactly supported.
Finally $|I_{2j}|\le p_{\psi_jf_j}(u)$, where
 $ p_{\psi_jf_j}(u)=|\langle u,\psi_jf_j\rangle|$,
and we obtain
\begin{eqnarray}\label{paireq}
|\langle u,v\rangle| & \le & \sum_j
\Big(
p_{\psi_jf_j}(u) + ||u||_{M,U_{\alpha j},\psi_j} C
I_n^{M-m} 
\nonumber\\&& + ||u||_{M,U_{\alpha j},\psi_j} 
   ||v||_{N,U_{\beta j},\psi_j} 
  I_n^{N+M} \Big).
\label{estuv}
\end{eqnarray}
Thus, $\langle u,v\rangle$ is well defined because all
the terms in the right hand side are finite
and the sum is over a finite number of $j$.
Note that $p_{\psi_j f_j}(u)$ and 
$||u||_{M,U_{\alpha j},\psi_j}$ are seminorms of 
$\calD'_\Gamma$ because $\psi_j f_j\in \calD(\Omega)$ and,
by construction, $U_{\alpha j}$ is a closed cone and 
$\supp\psi_j\times U_{\alpha j}\cap \Gamma=\emptyset$.
Equation~(\ref{estuv}) shows 
that, for any $v\in \calE'_\Lambda$, the map
 $u\mapsto\langle u,v\rangle$ is continuous.

The second case is $\Gamma=\dotT^*\Omega$ and $\Lambda=\emptyset$, so that
$\calD'_\Gamma=\calD'(\Omega)$ and $\calE'_\Lambda=\calD(\Omega)$.
The seminorm $|\langle u,v\rangle|=p_v(u)$ is then a seminorm
of $\calD'_\Gamma$ since $v\in \calD(\Omega)$. 
The last case is when $\Gamma=\emptyset$ and
$\Lambda=\dotT^*\Omega$, so that
$\calD'_\Gamma=C^\infty(\Omega)$ and $\calE'_\Lambda=\calE'(\Omega)$.
If we use the fact that the usual topology of $C^\infty(\Omega)$
is equivalent with the topology defined by $||\cdot||_{N,V,\chi}$
for all closed cones $V$ and all 
$\chi\in \calD(\Omega)$~\cite{BDH-13}, then we 
see that the elements of $\calE'(\Omega)$ are continuous
maps from $C^\infty(\Omega)$ to $\bbK$~\cite[p.~89]{Schwartz-66}.

Finally, the pairing is separating because,
if $\langle u,v\rangle=0$ for all $v\in \calE'_\Lambda$,
then $\langle u,f\rangle=0$ for all $f\in \calD(\Omega)$ because
$\calD(\Omega)\subset \calE'_\Lambda$  and a distribution $u$ which
is zero on $\calD(\Omega)$ is the zero distribution.
Similarly, $v=0$ if
$\langle u,v\rangle=0$ for all 
$u\in \calD'_\Gamma$
because $\calD(\Omega)\subset \calD'_\Gamma$.
\end{proof}

To simplify the discussion, we used Eskin's $\alpha_j$
and $\beta_j$ functions to build maps from
$v\in \calE'_\Lambda$ to $f_j \in C^\infty(\Omega)$. 
This can be improved by defining  maps from 
$\calE'_\Lambda$ to the Schwartz space $\calS$ of
rapidly decreasing functions 
(see section~\ref{nuclear-sect}).

\subsection{Normal space of distributions}
The usual spaces of  distribution theory (e.g. $\calD$,
$\calS$, $C^\infty$, $\calD'$, $\calS'$, $\calE'$), are 
\emph{normal spaces of distributions}~\cite[p.~10]{Schwartz-57},
which enjoy useful properties with respect to duality.
They are defined as follows:
\begin{dfn}
A Hausdorff locally convex space $E$ is said to be
a \emph{normal space of distributions} if there
are continuous injective linear maps $i:\calD(\Omega)\to E$
and $j:E\to \calD'(\Omega)$,
where $\calD'(\Omega)$ is equipped with its strong topology,
such that: (i) The image of $i$ is dense in $E$,
(ii) for any $f$ and $g$ in $\calD(\Omega)$
$\langle j\circ i (f),g\rangle=\int_\Omega 
f(x) g(x)dx$~\cite[p.~319]{Horvath}.
\end{dfn}
To transform $\calD'_\Gamma$ into a normal space of distributions
we need to refine its topology.
In the case of $\calD'_\Gamma$ condition (ii) is
obviously satisfied because the injections $i$ and $j$ are
the identity. The fact that $j$ is a continuous
injection means that the topology of $\calD'_\Gamma$
must be finer than the topology induced 
on it by the strong topology of $\calD'(\Omega)$~\cite[p.~302]{Treves}.
Therefore, we now equip $\calD'_\Gamma$ with 
the topology defined by the seminorms
$p_B(u)=\sup_{f\in B}|\langle u,f\rangle|$
of uniform convergence on the bounded sets $B$
of $\calD(\Omega)$ (instead of only the seminorms
$p_f=|\langle u,f\rangle|$) and we keep the seminorms
$||u||_{N,V,\chi}$ defined in section~\ref{wwk-sect}.
Since $p_B$ are the seminorms of $\calD'(\Omega)$,
$\calD'_\Gamma$ has more seminorms than $\calD'(\Omega)$,
the identity is a continuous injection
and its topology is finer than that of $\calD'(\Omega)$~\cite[p.~98]{Horvath}.
We call this topology the \emph{normal topology} of $\calD'_\Gamma$,
while the usual topology 
will be called the \emph{H\"ormander topology} of $\calD'_\Gamma$.
Note that $\calD'_\Gamma$ is Hausdorff for the normal
topology because it is Hausdorff for the coarser
H\"ormander topology.
It remains to show that 
\begin{lem}
\label{continjlem}
The injection of $\calD(\Omega)$ in $\calD'_\Gamma$ is continuous.
\end{lem}
\begin{proof}
We have to prove that the identity map
$\calD(\Omega)\hookrightarrow \calD'_\Gamma$ is continuous.
Because of the inductive limit topology of $\calD(\Omega)$,
we must show that, for any compact subset $K$ of $\Omega$,  the map
$\calD(K)\hookrightarrow \calD'_\Gamma$ is continuous
for the topology of $\calD(K)$~\cite[p.~66]{Schwartz-66}.
Recall that $\calD(K)$ is the set of elements of $\calD(\Omega)$
whose support is contained in $K$.
Its topology is defined by the seminorms
$\pi_{m,K}(f)=\sup_{|\alpha|\le m} \sup_{x\in K}
|\partial^\alpha f(x)|$.

Continuity is proved
by showing that all the seminorms of 
$\calD'_\Gamma$ are bounded by seminorms of 
$\calD(K)$~\cite[p.~98]{Horvath}.
Let $B$ be a bounded set of $\calD(\Omega)$ and
$p_B(f)=\sup_{g\in B} |\langle f,g\rangle|$
with 
$\langle f,g\rangle=\int_K f(x) g(x) dx$.
The function $f(x)$ is bounded by $\pi_{0,K}(f)$
and all the
$g(x)$ in $B$ are bounded by a common number $M_0$
because $B$ is bounded~\cite[p.~69]{Schwartz-66}. Thus,
$p_B(f)\le |K| M_0 \pi_{0,K}(f)$, where $|K|$
is the volume of $K$.

We still must estimate the seminorms
$||f||_{N,V,\chi}=\sup_{k\in V} (1+|k|)^N
|\widehat{f\chi}(k)|$.
By using
$(1+|k|) \le \beta (1+|k|^2)$, with
$\beta=(1+\sqrt{2})/2$, we find
\begin{eqnarray*}
(1+|k|)^N |\widehat{f\chi}(k)|
 & \le &
\beta^N \Big| (1+|k|^2)^N
\int e^{ik\cdot x} f(x)\chi(x) dx \Big|
\\&\le &
\beta^N \Big|
\int e^{ik\cdot x} (1-\Delta)^N (f\chi)(x) dx \Big|
\end{eqnarray*}
We expand $(1-\Delta)^N=\sum_{i=0}^N \binom{N}{i}
(-\Delta)^i$ and we estimate 
each $|\Delta^i (f\chi)(x)|\le n^i\pi_{2N,K}(f\chi)$.
This gives us
$(1+|k|)^N |\widehat{f\chi}(k)|\le ((1+n)\beta)^N |K|\pi_{2N,K}(f\chi)$.
To calculate $\pi_{2N,K}(f\chi)$ we notice that,
for any multi-index $\alpha$ such that $|\alpha|\le m$, we have
\begin{eqnarray}
|\partial^\alpha (f\chi)| &\le &
\sum_{\beta\le \alpha} \binom{\alpha}{\beta}
|\partial^\beta f||\partial^{\alpha-\beta}\chi|
\le 
\sum_{\beta\le \alpha} \binom{\alpha}{\beta}
\pi_{m,K}(f)\pi_{m,K}(\chi)
\nonumber\\& \le &
 2^m \pi_{m,K}(f)\pi_{m,K}(\chi).
\label{pimfchi}
\end{eqnarray}
Thus,
\begin{eqnarray}
(1+|k|)^N|\widehat{f\chi}(k)| &\le & (4(n+1)\beta)^N |K|
\pi_{2N,K}(\chi)\pi_{2N,K}(f),
\label{ekNVchi}
\end{eqnarray}
with a bound independent
of $k$ and
$||f||_{N,V,\chi}\le C \pi_{2N,K}(f)$,
where $C=(4(n+1)\beta)^N |K| \pi_{2N,K}(\chi)$.
The proof that the identity is continuous is complete.
\end{proof}

It is now clear that $\calD'_\Gamma$ with its normal topology
is a normal space of distribution because 
$\calD(\Omega)$ is dense in $\calD'_\Gamma$ (since sequential
convergence for the weak and strong topologies of $\calD'(\Omega)$
are equivalent~\cite[p.~70]{Schwartz-66} and 
from H\"ormander's density result~\cite[p.~262]{HormanderI}).
From the general properties of normal spaces of distributions
we obtain:
\begin{prop}
\label{normalprop}
If we (temporarily) denote by $\calD_\Gamma$ the dual 
of $\calD'_\Gamma$, then
\begin{itemize}
\item[(i)] The restriction map induces an injection $\calD_\Gamma\hookrightarrow\calD'(\Omega)$~\cite[p.~259]{Horvath}
\item[(ii)] If $\calD_\Gamma$ is equipped with 
the strong topology $\beta(\calD_\Gamma,\calD'_\Gamma)$,
then the injection $\calD_\Gamma\hookrightarrow \calD'(\Omega)$
is continuous~\cite[p.~259]{Horvath}
\item[(iii)] If $\calD_\Gamma$ is equipped with the 
topology $\kappa(\calD_\Gamma,\calD'_\Gamma)$
of uniform convergence on the balanced, convex, compact
sets for the normal topology of $\calD'_\Gamma$
(also called Arens topology~\cite{Arens-47}), then
$\calD_\Gamma$ is a normal space of distributions~\cite[p.~259]{Horvath}
and the dual of $\calD_\Gamma$ is $\calD'_\Gamma$~\cite[p.~235]{Horvath}
\item[(iv)] A distribution $v\in \calD'(\Omega)$ belongs to
$\calD_\Gamma$ if and only if it is continuous on
$\calD(\Omega)$ for the topology induced by 
$\calD'_\Gamma$~\cite[p.~319]{Horvath}
\item[(v)] $\calD(\Omega)$ is dense in $\calD_\Gamma$ equipped with
  any topology compatible with duality~\cite[p.~10]{Schwartz-57}
\end{itemize}
\end{prop}

We are now ready to prove
\begin{prop}
\label{firstdualprop}
The dual of $\calD'_\Gamma$ for its normal topology
is $\calE'_\Lambda$.
\end{prop}
\begin{proof}
We already proved that $\calE'_\Lambda \hookrightarrow\calD_\Gamma$
because, by lemma~\ref{pairinglem}, any $v\in\calE'_\Lambda$ defines
a continuous map $\calD'_\Gamma\to\bbK$ 
(for the H\"ormander and thus for the normal topology) 
and the injectivity is obvious by density of $\mathcal{D}(\Omega)$ in 
$\calD'_\Gamma$.
It remains to show that any continuous linear map 
$\lambda:\calD'_\Gamma\to\bbK$ defines a distribution
in $\calE'_\Lambda$. By item~(i) of proposition~\ref{normalprop},
we know that $\lambda$ is a distribution.
We first show that
this distribution is compactly supported, then that its
wavefront set is included in $\Lambda$.

Since the map $\lambda$ is continuous for the normal topology
of $\calD'_\Gamma$, there exists
a finite number of seminorms $p_i$ and a constant
$M$ such that $|\lambda(u)|\le M 
\sup_i p_i(u)$ for all $u$ in $\calD'_\Gamma$~\cite[p.~98]{Horvath}.
In other words, there is a bounded set $B$ in $\calD(\Omega)$ 
(one is enough because $\sup_i p_{B_i}\le p_B$ where
$B=\cup_i B_i$), and there are $r$ integers $N_i$,
$r$ functions $\chi_i$ in $\calD(\Omega)$ and $r$
closed cones 
$V_i$ such that
$\supp\chi_i\times V_i \cap \Gamma=\emptyset$ and
\begin{eqnarray*}
|\lambda(u)| \le M\sup(p_B(u),||u||_{N_1,V_1,\chi_1},
  \dots,||u||_{N_r,V_r,\chi_r}).
\end{eqnarray*}

We first show that $\lambda$ is a compactly supported
distribution. Indeed, $B$ is
a bounded set of $\calD(\Omega)$ if and only if there is a compact
subset $K$ of $\Omega$ and  constants $M_m$ such that
all $g\in B$ are supported on $K$ and 
$\pi_{m,K}(g)\le M_m$~\cite[p.~68]{Schwartz-66}.
According to the definition of the support of a
distribution~\cite[p.~42]{HormanderI}, 
$\langle u,g\rangle=0$ if 
$\supp u\cap \supp g=\emptyset$.
Thus $p_B(u)=\sup_{g\in B} |\langle u,g\rangle|=0$
if $\supp u\cap K=\emptyset$. Similarly,
$||u||_{N_i,V_i,\chi_i}=0$ if $\supp u\cap \supp \chi_i=
\emptyset$. Finally, for any $f\in \calD(\Omega)$ whose support
does not meet $K_\lambda=\cup_i \supp\chi_i \cup K$, we have
$|\lambda(f)|= 0$. This implies that the support
of $\lambda$ is included in the compact set 
$K_\lambda$~\cite[p.~42]{HormanderI}.

Then we show that
$\WF(\lambda)\subset \Lambda_M=\cup_{i=1}^M \supp\chi_i\times (-V_i)$.
We fix  an integer $N$,
a function $\psi\in \calD(\Omega)$ and a closed cone $W$
such that $\supp \psi\times W\cap \Lambda_M=\emptyset$
and we define $f_k=(1+|k|)^N \psi e_k$, where
$e_k(x)=e^{ik\cdot x}$. Hence,
\begin{eqnarray*}
||\lambda||_{N,W,\psi} &=& 
\sup_{k\in W} (1+|k|)^N |\widehat{\lambda\psi}(k)|
=\sup_{k\in W} |\lambda(f_k)|,
\end{eqnarray*}
where we used the fact that the Fourier transform of the
compactly supported distribution $\lambda\psi$
is $\lambda(\psi e_k)$~\cite[p.~165]{HormanderI}.
Since, by continuity, $|\lambda(f_k)|\le M\sup_i p_i(f_k)$, 
where $p_0=p_B$ and $p_i=||\cdot||_{N_i,V_i,\chi_i}$,
it suffices to bound each $\sup_{k\in W}p_i(f_k)$.

We first estimate $p_B(f_k)$. 
Since $B$ is a bounded set in $\calD(\Omega)$, the support of all $g\in B$ 
is contained in a common compact set $K$~\cite[p.~88]{Schwartz-66} and
\begin{eqnarray*}
|\langle f_k,g\rangle| &=& 
(1+|k|)^N |\langle \psi e_k,g\rangle|
= (1+|k|)^N |\widehat{\psi g}(k)| \\
& \le & (4(n+1)\beta)^N |K| \pi_{2N,K}(g)\pi_{2N,K}(\psi),
\end{eqnarray*}
where we used eq.~(\ref{ekNVchi}).
Moreover, all the seminorms of elements of $B$ are 
bounded~\cite[p.~88]{Schwartz-66}.
Thus, there is a number $M_{2N}$ such that
$\pi_{2N,K}(g)\le M_{2N}$ for all $g\in B$
and we obtain
$|\langle f_k,g\rangle| \le (8\beta)^N |K| \pi_{2N,K}(\psi)
M_{2N}$.  Since this bound is independent of $k$,
 we obtain our first bound 
$\sup_{k\in \mathbb{R}^n}p_B(f_k)<\infty$.

Consider now the second type of seminorms and 
calculate $p_i(f_k)=||f_k||_{N_i,V_i,\chi_i}$.
We have two cases:

(i) If $ (\supp\psi\cap\supp\chi_i)=\emptyset$,
then $\sup_{k\in \mathbb{R}^n}p_i(f_k)=0$ and we are done.

(ii) If $\supp\psi\cap\supp\chi_i\not=\emptyset$,
we want to estimate
\begin{eqnarray*}
||f_k||_{N_i,V_i,\chi_i} &=& 
\sup_{q\in V_i} (1+|q|)^{N_i} |\widehat{f_k\chi_i}(q)|=
\sup_{q\in V_i} (1+|q|)^{N_i}(1+|k|)^N |\widehat{e_k\psi\chi_i}(q)|.
\end{eqnarray*}
We have
$\widehat{e_k\psi\chi_i}(q)=\langle e_k e_q,\psi\chi_i\rangle
=\widehat{\psi\chi_i}(k+q)$.
Since we chose $W$ such that $(-V_i)\cap W=\emptyset$, 
by compactness of the intersection of $V_i$ and $W$ with the unit sphere, 
there is a $1\geq c>0$ such that $|k-q|/|k|>c$ and 
$|q-k|/|q|>c$ for all $k\in W$ and $q\in -V_i$.
We thus deduce:
\begin{eqnarray*}
||f_k||_{N_i,V_i,\chi_i} &\le &
 c^{-N-N_i} \sup_{q\in V_i}
  (1+|k+q|)^{N+N_i} \widehat{\psi\chi_i}(k+q).
\end{eqnarray*}
The function $\psi\chi_i$ is smooth and compactly supported.
We can use eq.~(\ref{ekNVchi}) again to show that
the right hand side of this inequality is bounded 
uniformly in $k$.

This concludes the proof of  
$\WF(\lambda)\subset \Lambda_M$.
Finally, $\supp\chi_i\times V_i\cap\Gamma=\emptyset$
implies $\supp\chi_i\times (-V_i)\subset\Lambda$
and $\Lambda_M\subset\Lambda$.
Thus, $\WF(\lambda)\subset \Lambda$ and 
since $\lambda$ is compactly
supported we have $\lambda\in \calE'_\Lambda$.
\end{proof}
In the following, we shall use $\calE'_\Lambda$
(instead of $\calD_\Gamma$) to denote the dual of 
$\calD'_\Gamma$.
Note that a similar proof shows that 
$\calE'_\Lambda$ is the topological dual of 
$\calD'_\Gamma$ equipped with the H\"ormander topology.
Indeed, lemma~\ref{pairinglem} shows in fact that the
pairing is continuous for the H\"ormander topology
because $p_{\psi_i f_i}$ in Eq.~\eqref{estuv} is a
seminorm of the weak topology of $\calD'(\Omega)$,
and the proof of the reverse inclusion just
requires to replace $p_B$ by a finite set of
$p_{f_i}$.

\section{Topologies on  $\calE'_\Lambda$}
Our purpose in this section is to show that,
if $(\calE'_\Lambda,\beta)$ denotes the space $\calE'_\Lambda$
equipped with the strong 
$\beta(\calE'_\Lambda,\calD'_\Gamma)$ topology,
then the topological dual of $(\calE'_\Lambda,\beta)$ 
 is
$\calD'_\Gamma$. This implies immediately that 
$\calD'_\Gamma$ is semi-reflexive 
and $\calE'_\Lambda$ is barrelled.
However, we shall not work directly with the strong
topology $\beta(\calE'_\Lambda,\calD'_\Gamma)$. It will
be convenient (especially to show that 
$\calE'_\Lambda$ is nuclear and $\calD'_\Gamma$ is complete) 
to define a topology on $\calE'_\Lambda$ as an inductive limit. 
Then, we prove that the inductive topology
is compatible with duality and we conclude by showing
that this inductive topology is equivalent to
the strong topology.

\subsection{Inductive limit topology on $\calE'_\Lambda$}
We want to define a topology on $\calE'_\Lambda$
as the topological inductive limit of some topological
spaces $E_\ell$. We shall first determine the vector
spaces $E_\ell$,  then we equip them with a topology.

Let us express $\calE'_\Lambda$ as the union
of increasing spaces  $E_\ell$.
Inspired by the work of Brunetti and coll.~\cite{Brunetti-09},
we take $E_\ell$ to be a set of distributions whose wavefront
set is contained in some closed cone, that we denote by $\Lambda_\ell$.
To determine $\Lambda_\ell$ we notice that $\Lambda$ is an open set and
the projection $\pi_i$ of a product
space into each of its coordinate spaces is
open~\cite[p.~90]{Kelley}. Thus, $\pi_1(\Lambda)$
is an open subset of $\Omega$. On the other hand,
the singular support of $v\in\calE'_\Lambda$ (i.e. 
$\Sigma(v)=\pi_1(\WF(v))$~\cite[p.~254]{HormanderI})
is closed~\cite[p.~108]{Friedlander}. It is even compact
because it is a closed subset of the support of 
$v$, which is compact. Hence, if we exhaust $\pi_1(\Lambda)$
by an increasing sequence of compact sets $K_\ell$ we know that, for any
$v\in \calE'_\Lambda$, $\Sigma(v)$ will be contained
in $K_\ell$ for $\ell$ large enough (because $\Sigma(v)\subset\pi_1(\Lambda)$
implies that the distance between the compact set $\Sigma(v)$ 
and the closed set $\pi_1(\Lambda)^c$ is strictly positive).
Let us define $K_\ell$ to be the set of points
that are at a distance smaller than $\ell$ from the origin
and at a distance larger than $1/\ell$ from the
boundary of $\Omega$ and from the boundary of 
$\pi_1(\Lambda)$:
$K_\ell=\{x\in \Omega\telque |x|\le \ell, d(x,\Omega^c)\ge 1/\ell,
d(x,\partial\pi_1(\Lambda))\ge 1/\ell\}$, where
$\partial\pi_1(\Lambda)$ is the boundary of $\pi_1(\Lambda)$ and
$d(x,A)=\inf\{|x-y|, y\in A\}$ is the distance between
a point $x$ and a subset $A$ of $\Omega$.
If $A$ is empty, we consider that $d(x,A)=+\infty$.
The sets $K_\ell$ are obviously compact
(they are intersections of closed sets with a compact ball),
$K_\ell\subset K_{\ell+1}$ and
$\pi_1(\Lambda)=\cup_{\ell=1}^\infty K_\ell$. Indeed,
$\Omega^c$ is closed because $\Omega$ is open and 
$\partial\pi_1(\Lambda)$ is a closed set disjoint from
$\pi_1(\Lambda)$ because $\pi_1(\Lambda)$ is open~\cite[p.~46]{Kelley}.
Thus, any point of $\pi_1(\Lambda)$ is at a finite distance 
$\epsilon_1$ from $\Omega^c$, $\epsilon_2$ from 
$\partial\pi_1(\Lambda)$ and $M$ from zero.
Then $x\in K_\ell$ for all integers $\ell$ greater
than $1/\epsilon_1$, $1/\epsilon_2$ and $M$.

We can now build the closed cones $\Lambda_\ell$, that will be subsets
of $\pi_1^{-1}(K_\ell)$ at a finite distance from $\Gamma'$:
$\Lambda_\ell=\{(x;k)\in \dotT^*\Omega\telque
x\in K_\ell, d\big((x;k/|k|),\Gamma'\big)\ge 1/\ell\}$.
This set is clearly a cone because it is defined in terms
of $k/|k|$ and it is closed in $\dotT^*\Omega$ because
it is the intersection of two close sets:
$\pi_1^{-1}(K_\ell)$ and $\{(x;k)\in \dotT^*\Omega\telque
d\big((x;k/|k|),\Gamma'\big)\ge 1/\ell\}$.
The first set is closed because
$K_\ell$ is compact and $\pi_1$ is continuous  and 
the second set is closed because the function
$(x;k)\mapsto d\big((x;k/|k|),\Gamma'\big)$ is continuous on
$\dotT^*\Omega$.

For some proofs, it will be useful for the support of the
distributions to be contained in a fixed compact set. Therefore,
we also consider an increasing sequence of 
compact sets $\{L_\ell\}_{\ell\in \bbN}$ exhausting $\Omega$
and such that $L_\ell$ is a compact neighborhood
of $K_\ell \cup L_{\ell-1}$ ($L_0=\emptyset$). Finally, we define 
$E_\ell=\calE'_{\Lambda_\ell}(L_\ell)$ to be the set
of distributions in $\calE'(\Omega)$ whose support is contained
in $L_\ell$ and whose wavefront set is contained in $\Lambda_\ell$. 
Note that $E_\ell$ will be equipped with the topology induced
by $\calD'_{\Lambda_\ell}$ as a closed subset
(it is closed because, by definition of the support
of a distribution, $E_\ell$ is the intersection of the
kernel of all continuous maps $u\mapsto \langle u,\phi\rangle$
where $\supp\phi\subset L_\ell^c$).

This is an increasing sequence of spaces exhausting
$\calE'_\Lambda$. It is increasing because
$L_\ell\subset L_{\ell+1}$ and 
$\Lambda_\ell\subset \Lambda_{\ell+1}$ imply
$\calE'_{\Lambda_\ell}(L_\ell)\subset
\calE'_{\Lambda_{\ell+1}}(L_{\ell+1})$.
To show that it is exhausting, consider
any $v\in \calE'_\Lambda$. 
Since the support of $v$ is compact, it is contained 
in some $L_{\ell_0}$ and then 
in $L_\ell$ for all $\ell \ge \ell_0$.
To show that $\WF(v)\subset \Lambda_{\ell_1}$
for some $\ell_1$,
consider the set $S_v=\{(x;k)\telque |k|=1\text{ and }
(x;k)\in WF(v)\}$. It is compact because it is closed
and bounded (the support of $v$ being compact).
Since $\WF(v)\subset \Lambda$ and $\Lambda\cap\Gamma'=\emptyset$,
we have $S_v\cap \Gamma'=\emptyset$.
There is a number $\delta>0$ such that
$d\big((x;k),\Gamma'\big) > \delta$ for all $(x;k)\in S_v$
because $S_v$ is compact and $\Gamma'$ is closed.
Thus, $S_v \subset \Lambda_\ell$ for $\ell > 1/\delta$.
Since both $S_v$ and $\Lambda_\ell$ are 
cones we have $\WF(v) \subset \Lambda_\ell$.
Finally, $v\in E_\ell$ for all $\ell$ larger
than $\ell_0$ and $1/\delta$.

We obtained the first part of
\begin{lem}
If $\Lambda$ is an open cone in $\dotT^*\Omega$, then
\begin{eqnarray*}
\calE'_\Lambda &=& \bigcup_{\ell=1}^\infty E_\ell,
\end{eqnarray*}
where $E_\ell=\calE'_{\Lambda_\ell}(L_\ell)$ is the 
set of distributions in $\calE'(\Omega)$
with a wavefront set contained in $\Lambda_\ell$
and a support contained in $L_\ell$.
If $E_\ell$ is equipped with the topology induced
by $\calD'_{\Lambda_\ell}$ (with its normal topology) we 
define on $\calE'_\Lambda$ the topological inductive limit
\begin{eqnarray*}
\calE'_\Lambda &=& \lim_{\to} E_\ell.
\end{eqnarray*}
This topology will be called the \emph{inductive topology}
on $\calE'_\Lambda$.
\end{lem}
\begin{proof}
The inductive limit of $E_\ell$ defines
a topology on $\calE'_\Lambda$
iff the injections $E_\ell\hookrightarrow E_{\ell+1}$
are continuous~\cite[p.~221]{Kothe-I}.
Since $E_\ell\subset \calD'_{\Lambda_\ell}$,
we can equip $E_\ell$ with the topology induced by
$\calD'_{\Lambda_\ell}$, which is 
defined by the seminorms $p_B(v)$ for all bounded sets
$B$ of $\calD(\Omega)$ and $||\cdot||_{N,V\chi}$, where 
$\supp\chi\times V\cap \Lambda_\ell=\emptyset$.
We prove that $E_\ell\hookrightarrow E_{\ell+1}$ is continuous
by showing that $E_\ell$ has more seminorms than $E_{\ell+1}$.
We have $\Lambda_\ell\subset\Lambda_{\ell+1}$.
Thus, $\Lambda_\ell^c \supset \Lambda_{\ell+1}^c$,
$\supp\chi\times V\cap \Lambda_\ell=\emptyset$
if $\supp\chi\times V\cap \Lambda_{\ell+1}=\emptyset$
and all the seminorms $||v||_{N,V,\chi}$ 
on $\calE'_{\Lambda_{\ell+1}}$ are also seminorms on
$\calE'_{\Lambda_\ell}$.
The seminorms $p_B$ are the same for 
$\calE'_{\Lambda_{\ell+1}}$ and
$\calE'_{\Lambda_\ell}$ because the sets $B$
are identical (i.e. the bounded sets of $\calD(\Omega)$).
\end{proof}

This inductive limit is not
strict if the open cone $\Lambda$ is not closed.
Indeed, if the inductive limit were strict, then the
Dieudonn\'e-Schwartz theorem~\cite[p.~161]{Horvath} would 
imply that each bounded set of $\calE'_\Lambda$ is included
and bounded in an $E_\ell$, which is wrong
when $\Lambda$ is not both open and closed, as we shall
prove in section~\ref{Ebornosect}.

\subsection{Duality of the inductive limit}
In this section, we show that the inductive topology on $\calE'_\Lambda$
is compatible with the pairing:
\begin{prop}
\label{inductprop}
The topological dual of $\calE'_\Lambda$ equipped with its
inductive topology is $\calD'_\Gamma$.
\end{prop}
\begin{proof}
We first show that $\calD'_\Gamma \hookrightarrow (\calE'_\Lambda)'$.
We already know that, for any $u\in \calD'_\Gamma$, 
$\langle u,v\rangle$ is well defined for all $v\in E_\ell$
because $E_\ell\subset \calE'_\Lambda$. 
Note that injectivity is obvious since smooth compactly supported 
functions, which form a separating set for 
distributions, are in $\calE'_\Lambda$.
A linear map from an inductive limit into a locally convex
space is continuous if and only if its restriction to all $E_\ell$ is
continuous~\cite[p.~217]{Kothe-I}.
Therefore, we must show that, for any
$\ell$, the map $\lambda:v\mapsto \langle u,v\rangle$ is continuous
from $E_\ell$ to $\bbK$.
The proof is so close to the derivation
of lemma~\ref{pairinglem} that it suffices to list
the differences.
We define a finite number of compactly supported
smooth functions $\psi_j$ such that 
$\sum_j \psi_j^2=1$ on a compact neighborhood of $L_\ell$
(here we use the fact that the support of all $v\in E_\ell$
is contained in a common compact set)
and closed cones $V_{uj}$ and
$V_{vj}$ satisfying the three conditions
(i) $V_{uj}\cap (-V_{vj})=\emptyset$,
(ii) $\supp\psi_j\times V_{uj}^c\cap \WF(u)=\emptyset$ and
(iii) $\supp\psi_j\times V_{vj}^c\cap \Lambda_\ell=\emptyset$.
The integral $I_{2j}$ is calculated as $I_{3j}$ in
lemma~\ref{pairinglem} if we interchange $u$ and $v$,
$\alpha$ and $\beta$:
$|I_{2j}|\le ||v||_{N,U_{\beta j},\psi_j} C I^{N-m}_n$,
where $m$ is the order of $v$,
and $I_{3j}$ is bounded as $I_{2j}$ in
lemma~\ref{pairinglem}:
$|I_{3j}|\le p_{\psi_jg_j}(v)$, where
$\hat{g}_j(k)=\beta_j(k)(1-\alpha_j(-k))\widehat{\psi_j u}(-k)$.
We obtain
\begin{eqnarray*}
|\langle u,v\rangle|  &\le & \sum_j\Big(
p_{\psi_jg_j}(v)+ ||v||_{N,U_{\beta j},\psi_j} C I^{N-m}_n
\\&&
+ ||u||_{M,U_{\alpha j},\psi_j} ||v||_{N,U_{\beta j},\psi_j} 
  I_n^{N+M}\Big),
\end{eqnarray*}
for any $N>m+n$ (the condition $N+M>n$ being then satisfied
for any nonnegative integer $M$).
This shows the continuity of $\lambda$ because the
right hand side is a finite sum of terms involving seminorms
of $\calD'_{\Lambda_\ell}$, which induce the topology of $E_\ell$.

Conversely, to prove that $(\calE'_\Lambda)'\hookrightarrow\calD'_\Gamma$,
we show that any element $\lambda$ of 
$(\calE'_\Lambda)'$ defines by restriction to $\calD(\Omega)$ a distribution and then that its wavefront
set is contained in $\Gamma$. This will be enough since by density of $\calD(\Omega)$ in $\calE'_\Lambda$ the restriction then extends uniquely to $\calE'_\Lambda$ and is thus the inverse of the reverse embedding. A linear map $\lambda:\calE'_\Lambda\to \bbK$
is continuous if its restriction to all
$E_\ell$ is continuous.
In other words, for each $E_\ell$ 
there is a bounded set $B$ in $\calD(\Omega)$
and there are smooth functions $\chi_i$ and
closed cones $V_i$ such that
$\supp\chi_i\times V_i \cap \Lambda_\ell=\emptyset$ and
\begin{eqnarray}
|\lambda(v)| \le M\sup(p_B(v),||v||_{N_1,V_1,\chi_1},
  \dots,||v||_{N_r,V_r,\chi_r}).
\label{alphav}
\end{eqnarray}
We first prove that $\lambda$ is a distribution, i.e. a 
continuous linear
map from $\calD(\Omega)$ to $\bbK$.
Recall that the space $\calD(\Omega)$
is the inductive limit of $\calD(L_\ell)$ 
because $L_\ell$ is an increasing sequence of
compact sets exhausting $\Omega$~\cite[p.~66]{Schwartz-66}. 
Thus, a map $\lambda$ is a distribution if the restriction of
$\lambda$ to each $\calD(K_\ell)$ is continuous.
For any $f\in \calD(K_\ell)$, we must show 
that all the seminorms on the right hand side
of eq.~(\ref{alphav}) can be bounded by some $\pi_m(f)$.
But this is a consequence of the fact that
$\calD(\Omega)\hookrightarrow \calD'_{\Lambda_\ell}$
is continuous, which was established in lemma~\ref{continjlem}.

Since $\lambda$ is a distribution, it has
a wavefront set. To prove that $\WF(\lambda)\subset\Gamma$ 
consider a smooth compactly
supported function $\psi$ and a closed cone
$W$ such that $\supp\psi\times W \cap \Gamma=\emptyset$,
i.e.  $\supp\psi\times (-W)\subset \Lambda$.
Since the restriction of $\supp\psi\times (-W)$
to the unit sphere is compact, there is an $\ell$
such that $\supp\psi\times (-W)\subset \Lambda_\ell$.
Note also that $\supp\psi \subset \pi_1(\Lambda_\ell)
\subset L_\ell$ so that $f_k=(1+|k|)^N \psi e_k$ is in $E_\ell$.
We can now repeat the same reasoning as for the proof 
of proposition~\ref{firstdualprop} to show that
$||\lambda||_{N,W,\psi}=\sup_{k\in W}|\lambda(f_k)|$ is bounded.
This shows that $\WF(\lambda)\subset \Gamma$, which implies
$\lambda\in \calD'_\Gamma$ and
$(\calE'_\Lambda)'\subset \calD'_\Gamma$.

This completes the proof that $(\calE'_\Lambda)'=\calD'_\Gamma$.
\end{proof}

\subsection{The strong topology on $\calE'_\Lambda$}
We showed that the coupling between 
$\calE'_\Lambda$ and $\calD'_\Gamma$ is compatible
with duality. Thus, the inductive topology on $\calE'_\Lambda$
is coarser than the Mackey topology~\cite[p.~IV.4]{Bourbaki-TVS}.
The strong topology $\beta(\calE'_\Lambda,\calD'_\Gamma)$ is always
finer than the Mackey topology~\cite[p.~IV.4]{Bourbaki-TVS}.
Therefore, if we can show that the inductive
topology is finer than the strong topology,
we prove the identity of the inductive, Mackey and strong topologies.

\begin{lem}
\label{stronglem}
The inductive, Mackey and strong topologies on
$\calE'_\Lambda$ are equivalent.
\end{lem}
\begin{proof}
To show that the identity map, from
$\calE'_\Lambda$ with the inductive topology to
$\calE'_\Lambda$ with the strong topology, is continuous,
we must prove that the identity map is continuous
from all $E_\ell$  to 
$\calE'_\Lambda$ with the strong topology.
In other words, for any bounded set
$B'$ of $\calD'_\Gamma$, we must show that
$p_{B'}(v)=\sup_{u\in B'}|\langle u,v\rangle|$ is bounded on 
$E_\ell$ by some seminorms of $E_\ell$.

We proceed as in the proof of lemma~\ref{pairinglem}.
From the fact that $\Gamma'\cap \Lambda_\ell=\emptyset$
and $\supp v\subset L_\ell$ we can build a finite number
of smooth compactly
supported functions $\psi_j$ such that $\sum_j \psi_j^2=1$
on  a compact neighborhood $K'$ of $L_\ell$, and closed cones
$V_{uj}$ and $V_{vj}$ satisfying the three conditions
(i) $V_{uj}\cap (-V_{vj})=\emptyset$,
(ii) $\supp\psi_j\times V_{uj}^c\cap \Gamma=\emptyset$ and
(iii) $\supp\psi_j\times V_{vj}^c\cap \Lambda_\ell=\emptyset$.
The support of all $\psi_j$ is assumed to be contained
in a common compact neighborhood $K$ of $K'$.
Then, we define again homogeneous functions $\alpha_j$
and $\beta_j$ of degree 0, measurable, smooth except 
at the origin, non-negative and bounded by 1 on $\bbR^n$, such that 
the closed cones $\supp\alpha_j$
and $\supp\beta_j$ satisfy the three conditions (i), (ii)
and (iii), with
$\alpha_j=1$ on $V_{uj}$ and $\beta_j=1$ on $V_{vj}$
and, as in the proof of lemma~\ref{pairinglem}, we write
$\langle u, v\rangle=\sum_j (I_{1j}+I_{2j}+ I_{3j}+I_{4j})$.
We have again $I_{1j}=0$ because the supports of
$\alpha_j$ and $\beta_j$ are disjoint, 
and $|I_{4j}|\le ||u||_{M,U_{\alpha j},\psi_j} 
  ||v||_{N,U_{\beta j},\psi_j} 
  I_n^{N+M}$ for any integers $N$ and $M$  such that $N+M>n$.  
It is important to remark that $\psi_j$,
$\alpha_j$ and $\beta_j$ depend only on $\Gamma$,
$L_\ell$ and $\Lambda_\ell$ and not on $u$ and $v$.

To estimate $I_{2j}$ and $I_{3j}$, we need to establish some
properties of the bounded sets of $\calD'_\Gamma$.
The continuity of the injection
$\calD'_\Gamma\hookrightarrow \calD'(\Omega)$ 
implies that a set $B'$ which is bounded in $\calD'_\Gamma$
is also bounded in $\calD'(\Omega)$~\cite[p.~109]{Horvath}.
According to Schwartz~\cite[p.~86]{Schwartz-66},
a subset $B'$ is bounded in $\calD'(\Omega)$
iff, for any relatively compact open set $U\subset\Omega$,
there is an integer $m$ such that every $u\in B'$
can be expressed in $U$ as $u=\partial^\alpha f_u$
for $|\alpha|\le m$, where $f_u$ a continuous function.
Moreover, there is a number $M$
such that $|f_u(x)|\le M$ for all $x\in U$
and $u\in B'$. 
The elements of $E_\ell$ are supported on $L_\ell$
and we need only consider bounded sets of $\calD'_\Gamma$
that are defined on the compact neighborhood
$K$ of $L_\ell$. Thus, we can take for $U$ any
relatively compact open set containing $K$.

To calculate $I_{2j}$, as in the proof of lemma~\ref{pairinglem},
we define $\hat{g}_j(k)=\alpha_j(-k)(1-\beta_j(k))
\widehat{\psi_j v}(k)$ and we obtain
$I_{2j}=(2\pi)^{-n} \int \widehat{\psi_ju}(-k) \hat{g}_j(k) dk
=\langle u,\psi_j g_j\rangle$. At this stage one might apply
the Banach-Steinhaus theorem but we shall use an equivalent
method using $u=\partial^\alpha f_u$:
\begin{eqnarray*}
\langle u,\psi_j g_j\rangle &=&
\langle \partial^\alpha f_u,\psi_j g_j\rangle=
(-1)^{|\alpha|} \langle f_u,\partial^\alpha (\psi_jg_j)\rangle
=
(-1)^{|\alpha|} \langle \varphi f_u,\partial^\alpha (\psi_jg_j)\rangle,
\end{eqnarray*}
where $\varphi$ is a smooth function, equal to 1 on $K$ and
supported on $U$. Thus
\begin{eqnarray*}
\langle u,\psi_j g_j\rangle &=&
(-1)^{|\alpha|} (2\pi)^{-n} \int_{\bbR^n} \widehat{\varphi f_u}(-k)
\widehat{\partial^\alpha (\psi_jg_j)}(k) dk
\\ &=&
i^{|\alpha|} (2\pi)^{-n} \int_{\bbR^n} \widehat{\varphi f_u}(-k)
k^\alpha \widehat{\psi_jg_j}(k) dk.
\end{eqnarray*}
We must estimate
$\widehat{\psi_jg_j}(k)=(2\pi)^{-n}
\int_{\bbR^n} \widehat{\psi_j}(k-q)
\widehat{g_j}(q) dq$.
The functions $\alpha_j$ and $(1-\beta_j)$ are bounded
by 1 and $\widehat{\psi_jv}$ is fast decreasing on
$U_{\beta j}=\supp(1-\beta_j)$.
Thus,
$|\widehat{g_j}(q)|\le ||v||_{N,U_{\beta j},\psi_j}
(1+|q|)^{-N}$ for all integers $N$.
In the proof of lemma~\ref{continjlem}, we estimated
the Fourier transform of a smooth compactly supported
function: 
$|\psi_j(k-q)|\le C^{N'}_j (1+|k-q|)^{-N'}$ for all integers $N'$,
where $C^{N'}_j=((1+n)\beta)^{N'} |K| \pi_{2N',K}(\psi_j)$.
If we take $N=n+m+1$, where $m=|\alpha|$ is the
degree of $\partial^\alpha$,  and $N'=2N$ we obtain
\begin{eqnarray*}
|\widehat{\psi_jg_j}(k)| &\le &
(2\pi)^{-n} ||v||_{N,U_{\beta j},\psi_j} C^{2N}_j
\int_{\bbR^n} (1+|k-q|)^{-2N} (1+|q|)^{-N} dq
\\ &\le& ||v||_{N,U_{\beta j},\psi_j} C^{2N}_j
I_n^N (1+|k|)^{-N},
\end{eqnarray*}
where we used
$(1+|q|)^{-N} \le (1+|k-q|)^N (1+|k|)^{-N}$~\cite[p.~50]{Eskin}.
This estimate enables us to calculate
\begin{eqnarray*}
|I_{2j}| &=& |\langle u,\psi_j g_j\rangle| \le
(2\pi)^{-n} \int_{\bbR^n} |\widehat{\varphi f_u}(-k)|
|k|^m |\widehat{\psi_jg_j}(k)| dk
\\
&\le&
(2\pi)^{-n}  |U| M
||v||_{N,U_{\beta j},\psi_j} C^{2N}_j
I_n^N 
  \int_{\bbR^n} 
\frac{|k|^m}{ (1+|k|)^{n+m+1}}dk
\\ &\le &
  |U| M
||v||_{N,U_{\beta j},\psi_j} C^{2N}_j
I_n^N I_n^{n+1},
\end{eqnarray*}
where $N=n+m+1$, $|U|$ is the volume of $U$ and we used
the obvious bound
$|\widehat{\varphi f_u}(-k)| \le  |U| M$.

For the estimate of $I_{3j}$ we
start from 
$I_{3j}=(2\pi)^{-n} \int \widehat{g_j^u}(-k) 
  \widehat{\psi_j v}(k) dk$,
where $\widehat{g_j^u}(-k)=(1-\alpha_j(-k))\beta_j(k)
\widehat{\psi_j u}(-k)$.
Thus $|I_{3j}|=|\langle \psi_j g_j^u,v\rangle|$
can be bounded by $p_{B_j}(v)=\sup_{f\in B_j}|\langle f,v\rangle|$ if
the set $B_j=\{\psi_j g_j^u\telque u\in B'\}$ is
a bounded set in $\calD(\Omega)$.
It is clear that all $f\in B_j$ are supported on
$K=\supp\psi_j$ and that all $\psi_j g_j^u$
are smooth because $\psi_j$ is smooth and the
Fourier transform of $g_j^u$ is fast decreasing.
It remains to show that all the derivatives of
$\psi_j g_j^u$ are bounded by a constant independent
of $u$. For this we write
\begin{eqnarray*}
\partial^\gamma (\psi_j g_j^u)(x) &=& 
(2\pi)^{-n} (-i)^{|\gamma|} \int_{\bbR^n}
  e^{-ik\cdot x} k^\gamma \widehat{\psi_j g_j^u}(k) dk.
\end{eqnarray*}
If $|\gamma|\le m$, we use the estimate of 
$\widehat{\psi_j g_j}$ obtained in the previous
section and we interchange $u$ and $v$, $\alpha_j$
and $\beta_j$
\begin{eqnarray*}
|\widehat{\psi_jg^u_j}(k)| &\le &
||u||_{N,U_{\alpha j},\psi_j} C^{2N}_j
I_n^N (1+|k|)^{-N},
\end{eqnarray*}
where $N=n+m+1$. A set $B'$ is bounded in $\calD'_\Gamma$
iff it is bounded for all the seminorms of 
$\calD'_\Gamma$~\cite[p.~109]{Horvath}.
In particular, there is a constant
$M_{N,U_{\alpha j},\psi_j}$ such that
$||u||_{N,U_{\alpha j},\psi_j} \le
M_{N,U_{\alpha j},\psi_j}$ for all $u\in B'$.
Thus, for all $f\in B_j$,
$|\widehat{f}(k)| \le 
M_{N,U_{\alpha j},\psi_j} C^{2N}_j
I_n^N (1+|k|)^{-N}$ and
\begin{eqnarray*}
|\partial^\gamma f(x)| &\le& 
(2\pi)^{-n} \int_{\bbR^n} |k|^m |\widehat{f}(k)| dk
\le 
M_{N,U_{\alpha j},\psi_j} C^{2N}_j
I_n^N I_n^{n+1},
\end{eqnarray*}
where $N=n+m+1$ as for the estimate of $I_{3j}$.
In other words, for any $\gamma$ 
there is a constant $C_{|\gamma|}$ such that
$|\partial^\gamma f|\le C_{|\gamma|}$
for all $f\in B_j$. Thus,
$\pi_m(f)\le \sup_{0\le k\le m} C_k$ is bounded
independently of $f$, and we proved that
$B_j$ is a bounded set of $\calD(\Omega)$.
Hence, $|I_{3j}| \le p_{B_j}(v)$ where $p_{B_j}$
is a seminorm of $\calD'_\Gamma$.

If we gather our results we obtain
\begin{eqnarray}
p_{B'}(v) & \le &
\sum_j\Big( 
M_{n,U_{\alpha j},\psi_j} 
  ||v||_{n,U_{\beta j},\psi_j} 
  I_n^{2n}
+
 M 
||v||_{N,U_{\beta j},\psi_j} C^{2N}_j
I_n^N I_n^{n+1}
\nonumber\\&&
+p_{B_j}(v)
\Big),
\end{eqnarray}
where the sum over $j$ is finite and $N=n+m+1$ where
$m$ is the maximum order of the distributions 
of $B'$. The proof is complete.
\end{proof}

\section{Nuclearity}
\label{nuclear-sect}
In this section we investigate the nuclear
properties of the spaces studied in this paper.
To prove that $\calD'_\Gamma$ with the normal
topology is nuclear, we use a theorem
due to Grothendieck~\cite[Ch.~II, p.~48]{Grothendieck-55}
that can be expressed as follows~\cite[p.~92]{Pietsch}:
\begin{thm}
Let $E$ be a locally convex space and $(f_i)_{i\in I}$
a family of continuous linear maps from $E$ to
nuclear locally convex spaces $F_i$. If the topology
of $E$ is the initial topology for the maps $f_i$,
then $E$ is nuclear.
\end{thm}
We recall that, if the topology of $F_i$ is defined by
the seminorms $(p^i_\alpha)_{\alpha\in J_i}$, then
the initial topology of $E$ is defined by the seminorms
$(p^i_{\alpha_i}\circ 
f_i)_{i\in I,\alpha_i\in J_i}$~\cite[p.~152]{Horvath}.

The simplest case to prove is
\begin{prop}
The space $\calD'_\Gamma$ with the normal topology is
nuclear.
\end{prop}
\begin{proof}
We first construct the spaces $F_i$ and the linear maps $f_i$.
For $i=0$ we take
$F_0=\calD'(\Omega)$ and $f_0$ the continuous inclusion
$\calD'_\Gamma\hookrightarrow \calD'(\Omega)$
where $\calD'(\Omega)$, equipped with its strong
topology, is nuclear~\cite[p.~53]{Treves}.
For each $i=(V,\chi)$ where 
$(\supp\chi\times V)\cap\Gamma=\emptyset$,
the target space $F_i$ will be the Schwartz
space $\calS$ of rapidly decreasing functions on 
$\bbR^n$ equipped
with the family of seminorms~\cite[p.~90]{Horvath}.
\begin{eqnarray*}
||f||_{N,m} &=& \sup_{|\alpha|\le m}\sup_{k\in\bbR^n} 
  (1+|k|)^N |\partial^\alpha f(k)|.
\end{eqnarray*}
The space $\calS$ is nuclear~\cite[p.~430]{Treves}.
To build the linear maps, we choose a real function
$h\in \calD(\bbR^n)$ such that $h(k)=1$ for $|k|\le 1$,
$h(k)=0$ for $|k|>2$ and $0\le h(k) \le 1$ for all $k$, 
and a nonnegative function $\gamma\in \calD(\bbR^n)$ which is
bounded by 1, equal to 1 on $V\cap S^{n-1}$ and such that
$(\supp\chi\times\supp\gamma) \cap \Gamma=\emptyset$.
We define the homogeneous function $\zeta(k)=\gamma(k/|k|)$,
which is smooth outside the origin and bounded by 1.
The function $g=(1-h)\zeta$ is smooth on $\bbR^n$.
By using the homogeneity of $\zeta$ and the fact that
$h$ and $\gamma$ are in $\calD(\bbR^n)$, we see that
for any integer $m$ there is a constant $C_m$ such that
$|\partial^\alpha g(k)|\le C_m $ for
all $|\alpha|\le m$.

We can now define
$f_i:\calD'_\Gamma \to \calS$ by
$f_i(u) = g\,\widehat{u\chi}$.
The functions $f_i(u)$ are in $\calS$
because $\widehat{u\chi}$ is in $\calS$
by definition of the wavefront set and
$g$ is supported on the cone
$W=\{\lambda k\telque k\in \supp\gamma, |k|=1, \lambda>0\}$
and $(\supp\chi\times W)\cap\Gamma=\emptyset$.
To show that $f_i$ is continuous, we
must estimate $||f_i(u)||_{N,m}$ in
terms of the seminorms of $u$ in $\calD'_\Gamma$.
By noticing that
$\partial^\alpha \widehat{u\chi}=
\widehat{(ix)^\alpha u\chi}$ we obtain
\begin{eqnarray*}
||f_i(u)||_{N,m} &\le &\sup_{k\in W} (1+|k|)^N
  \sup_{|\alpha|\le m} \sum_{\beta}
 \binom{\alpha}{\beta}   C_m 
  |\widehat{x^\beta u\chi}(k)|
  \le  C \sup_{|\alpha|\le m} ||u||_{N,W,x^\alpha\chi}.
\end{eqnarray*}

We have shown that all $f_i$ are continuous. Thus,
the topology of $\calD'_\Gamma$ is finer than the
initial topology defined by the family $f_i$.
To show that the two topologies are equivalent,
it remains to prove that every seminorm defining the
topology of $\calD'_\Gamma$ can be bounded with
seminorms of the initial topology.

This is obvious for the seminorms of $\calD'(\Omega)$
because they are the same in $\calD'_\Gamma$.
For the seminorms $||\cdot||_{N,V,\chi}$ we note that
$\widehat{u\chi}=h\widehat{u\chi} +(1-h)\widehat{u\chi}$.
The function $g=(1-h)\zeta$ corresponding to $i=(V,\chi)$
enables us to write
$\widehat{u\chi}=h\widehat{u\chi} +g\widehat{u\chi}
=h\widehat{u\chi} +f_i(u)$
on $V$ and we obtain
\begin{eqnarray}
||u||_{N,V,\chi} & \le &
||f_i(u)||_{N,0} + \sup_{k\in V} (1+|k|)^N|h(k)\widehat{u\chi}(k)|.
\label{nucbound}
\end{eqnarray}
We just need a bound for the last term.
We notice that
$\widehat{u\chi}(k)=\langle u,\chi e_k\rangle$,
where $e_k(x)=e^{ik\cdot x}$,
so that
$ \sup_{k\in V} (1+|k|)^N|h(k)\widehat{u\chi}(k)| \le
p_B(u)$, where $B=\{(1+|k|)^N  \chi e_k\telque k\in V\cap \text{supp}(h)\}$.
Thus, the equivalence is proved if
$p_B$ is a seminorm of the strong topology
of $\calD'(\Omega)$, i.e. if $B$ is bounded in $\calD(\Omega)$.
All the elements of $B$ are supported on $K=\supp \chi$.
It remains to show that they are bounded
for all seminorms $\pi_{m,K}$ but this
is obvious by 
Eq.~\eqref{pimfchi} and 
$\pi_{m,K}(e_k)\le  |k|^m$.
\end{proof}

We emphasize an interesting structural consequence of the proof above for $\calD'_\Gamma$. Recall that the class of (PLS)-spaces is the smallest class stable by countable projective limits and containing strong duals of Fr\'echet-Schwartz spaces. Since such strong duals are known to be inductive limits of Banach spaces with compact linking maps, they are also called (LS)-spaces  and since they are bornological, their associated convex bornological space is sometimes called a Silva space \cite{Hogbe-77,Hogbe-81}.
This class appeared  recently as useful in applications
of homological algebra to functional analysis (see e.g. \cite{Wengenroth}) having
applications to parameter dependence of PDE's \cite{Domanski2}. It is known that any
Fr\'echet-Schwartz space is a (PLS)-space. See more generally \cite{Domanski} for a
review. It is also known that the strong dual of a (PLS)-space is an (LFS)-space (see below), i.e. a countable inductive limit of Fr\'echet-Schwartz spaces. Moreover, both are well-known to be strictly webbed spaces in the sense of De Wilde (using general stability properties of these spaces, see e.g. \cite[\S 35]{Kothe-II}) and thus they satisfy corresponding open-mapping and closed graph theorems. Recall also that the classical sequence space $\mathfrak{s}$ 
is known to be
isomorphic to the Fr\'echet nuclear space $\mathcal{S}$ 
(see e.g.~\cite[pp.~325 and~413]{Valdivia})
 having universal properties for nuclear spaces 
in the sense that any nuclear locally convex space is a linear subspace of 
$\mathfrak{s}^I$ for some set $I$.
\begin{cor}\label{PLS}
$\calD'_\Gamma$ with its normal topology is isomorphic to a closed subspace of the countable product  $(\mathfrak{s}')^\bbN\times (\mathfrak{s})^\bbN$, and thus it is a (PLS)-space and its strong  dual $\calE'_\Lambda$ is an (LFS)-space.
\end{cor}
\begin{proof}
By \cite[p. 385]{Valdivia} $\calD'(\Omega)$ is known to be isomorphic to 
$(\mathfrak{s}')^\bbN$. Moreover, we showed in ref.~\cite{BDH-13} (see alternatively 
~\cite[p.~80]{Grigis}) that the additional
seminorms of $\calD'_\Gamma(\Omega)$ could be chosen in a
countable set $\{p_n\telque n\in \bbN\}$. Thus, the proof of our previous lemma 
gives an embedding of $\calD'_\Gamma(\Omega)$ in $(\mathfrak{s}')^\bbN\times 
(\mathfrak{s})^\bbN$. Finally, one can either prove directly that this subspace is 
closed (and deduce in this way that $\calD'_\Gamma(\Omega)$ is complete) or merely 
use the completeness of $\calD'_\Gamma(\Omega)$ proved below in 
Corollary~\ref{completecor} to deduce that
it is necessarily closed as any complete 
subspace of a Hausdorff space. Finally, it is known 
(see e.g.~\cite[p.~96]{Wengenroth}) that a closed subspace of a (PLS)-space 
is again a (PLS)-space. The fact that the dual is an (LFS)
space is also well-known but we recall the argument by lack of an explicit 
reference. Since a complete Schwartz space is semi-Montel, its strong dual 
is also its Mackey dual, since closed subspaces of (LS) spaces are still 
(LS) spaces, we can assume the projective limit of (LS)-spaces to be reduced, 
so that one can apply \cite[\S 22.7.(9) p 294]{Kothe-I} to get its 
Mackey dual as an inductive limit of Mackey duals. But an (LS) space 
is known to be a Montel space thus this Mackey dual is also its 
strong dual which is known to be a Fr\'echet-Schwartz space 
(see e.g. \cite[Prop 8.5.26 p 293]{Perez-Carreras} or \cite[p 28]{Hogbe-81}).

\end{proof}

The fact that $\calD'_\Gamma$ is also nuclear for
the H\"ormander topology was stated by 
Fredenhagen and Rejzner~\cite{Fredenhagen-11}. However,
since the proof was only sketched, we demonstrate it
for completeness.
\begin{prop}
The space $\calD'_\Gamma$ with the H\"ormander topology is
nuclear.
\end{prop}
\begin{proof}
The map $f_0: \calD'_\Gamma\to \calD'(\Omega)$ goes now from
$\calD'_\Gamma$ with the H\"ormander topology to
$\calD'(\Omega)$ with the weak topology, which 
is also nuclear (every locally convex space being
nuclear for its weak topology~\cite[p.~202]{Hogbe-81}).
The end of the previous proof cannot be used because
the seminorm $p_B$ is not available in the weak topology.
Instead we define, for each $j=(K,\chi)$ where $K$ is 
the image of $[0,1]^n$ by an invertible linear map  $L$ such that 
$(\supp\chi\times K) \cap \Gamma=\emptyset$, the additional
map $g_j:\calD'_\Gamma\to C^\infty(K)$,
where $C^\infty(K)$ is 
the space of functions $f\in C^\infty(\mathring{K})$
such that $f$ and all its derivatives have continuous extensions
to $K$. The space $C^\infty(K)$, equipped with the seminorms
$\pi_{m,K}$,  is a nuclear space because 
$K=L([0,1]^n)$, where
$L$ is a linear change of variable, and $C^\infty([0,1]^n)$ is 
nuclear~\cite[pp.~325 and 378]{Valdivia} or \cite{Ogrodzka}.

 We define  $g_j(u)=h\,\widehat{u\chi}|_K$
(i.e. the restriction to $K$ of the smooth function
$h\widehat{u\chi}$).
The maps $g_j$ are continuous because
$\pi_{m,K}(h\widehat{u\chi})  \le 
  2^m \pi_{m,K} (h) \pi_{m,K}(\widehat{u\chi})$
and $\pi_{m,K}(\widehat{u\chi})\le \sup_{|\alpha|\le m}
   ||u||_{0,V,x^\alpha \chi}$ with $V=\mathbb{R}_+K.$
Conversely, for $V$ a closed cone such that 
$(\supp\chi\times V)\cap\Gamma=\emptyset$, 
there is a finite set of $K_\ell=L_\ell([0,1]^n)$
such that $(\supp h\cap V)\subset \cup_{\ell=1}^p
\mathring{K}_\ell$ and $(\supp\chi\times K_\ell)\cap\Gamma=\emptyset$.
Indeed, for every $k\in (\supp h\cap V)$, there is parallelepiped
$K_k=L_k([0,1]^n)$, with one vertex at zero, such that
$k\in \mathring{K}_k$ and $(\supp\chi\times
K_k)\cap\Gamma=\emptyset$. Thus, 
$(\supp h\cap V)\subset \cup_k \mathring{K}_k$ and we can
extract a finite covering because $\supp h\cap V$ is compact.
To estimate $(1+|k|)^N|h(k)\widehat{u\chi}(k)|$
in the right hand side of inequality~\eqref{nucbound},
we can take $|k|\le 2$ because $h(k)=0$ for $|k|>2$
and, for every $k\in \supp h\cap V$, we have
$|h\widehat{u\chi}(k)|=g_\ell(u)(k)$ if $k\in K_\ell$
and 
$|h\widehat{u\chi}(k)|=0$ if $k\notin K_\ell$,
where $g_\ell(u)=|h\widehat{u\chi}|_{K_\ell}$.
Thus, for all $k\in V$,
\begin{eqnarray*}
(1+|k|)^N|h(k)\widehat{u\chi}(k)|
 & \le &
 3^N \max_{\ell=1,...,p}
   \pi_{0,K_\ell}(g_{j_\ell}(u)),
\end{eqnarray*}
and
\begin{eqnarray*}
||u||_{N,V,\chi} & \le &
||f_i(u)||_{N,0} + 3^N \max_{\ell=1,...,p}\left( 
   \pi_{0,K_\ell}(g_{j_\ell}(u))\right).
\end{eqnarray*}
Thus, the H\"ormander topology is nuclear because it is the initial
topology of $(f_i)$ and $(g_j)$.
\end{proof}

To complete this section, we show that
\begin{prop}
The space $\calE'_\Lambda$ with the strong topology is
nuclear.
\end{prop}
\begin{proof}
Each $E_\ell$ is nuclear because it is
a vector subspace of the nuclear
space $\calD'_{\Lambda_\ell}$ with the normal
topology~\cite[p.~514]{Treves}.
Thus, $\calE'_\Lambda$ is nuclear since it is the countable
inductive limit of the nuclear spaces $E_\ell$~\cite[p.~514]{Treves}.
\end{proof}

\section{Bornological properties}
\label{borno-sect}

We study the bornological properties of $\calD'_\Gamma$ 
because they enable us to prove that $\calD'_\Gamma$ is complete
and because they have a better behaviour than the topological
properties with respect to the tensor product of sections.
More precisely, if $\Gamma_c(E)$ is the space of
compactly supported sections of a vector bundle
$E$ over $M$, then there is a bornological isomorphism between 
$\Gamma_c(E\otimes F)$ and $\Gamma_c(E)\otimes_{C^\infty(M)}^\beta\Gamma(F)$,
where $F$ is another vector bundle over $M$~\cite{Nigsch-13}.
As a consequence, there is also a bornological isomorphism
between the distribution spaces $\Gamma_c(E\otimes F)'$ and 
$\Gamma(E^*)\otimes_{C^\infty(M)}^{\beta} \Gamma_c(F)'$~\cite{Nigsch-13}.

\subsection{Bornological concepts}
We start by recalling some elementary concepts of
bornology theory~\cite{Hogbe-77}.
\begin{dfn}
A \emph{bornology} on a set $X$ is a family $\calB$ of subsets
of $X$ satisfying the following axioms:
\begin{itemize}
\item [B.1:] $\calB$ is a covering of $X$, i.e.
  $X=\cup_{B\in \calB} B$.
\item [B.2:] $\calB$ is hereditary under inclusion:
   if $A\in \calB$ and $B\subset A$, then $B\in \calB$.
\item [B.3:] $\calB$ is stable under finite union.
\end{itemize}
\end{dfn}
A pair $(X,\calB)$ is called a \emph{bornological set} and
the elements of $\calB$ are called the \emph{bounded subsets}
(or the bounded sets) of $X$. 


To define a convex bornological space we need
the concept of a \emph{disked hull}~\cite[p.~6]{Hogbe-77}.
We recall that a subset $A$ of a vector space is
a \emph{disk} if it is convex 
and balanced (i.e. if $x\in A$ and $\lambda\in \bbK$
with $|\lambda|\le 1$, then $\lambda 
x\in A$)~\cite[p.~4]{Hogbe-77}.

\begin{dfn}
If $E$ is a vector space, the disked hull
of a subset $A$ of $E$, denoted by $\Gamma(A)$, is
the smallest disk containing $A$.
\end{dfn}

\begin{dfn}
Let $E$ be a vector space on $\bbK$. A bornology
$\calB$ on $E$ is said to be 
a \emph{convex bornology}
if, for every $A$ and $B$ in $\calB$ and every
$t\in \bbK$, the
sets $A+B$, $tA$ and $\Gamma(A)$ 
belong to $\calB$.
Then $E$ or $(E,\calB)$ is called a convex bornological space.
\end{dfn}

We shall also need to define the convergence of
a sequence in a convex bornological space ~\cite[p.~12]{KrieglMichor}:
\begin{dfn}
Let $E$ be a convex bornological space. A sequence $x_n$ in
$E$ is said to \emph{Mackey-converge} to $x$ if there
exist a disked bounded subset $B$ of $E$ and a sequence
$\alpha_n$ of positive real numbers tending to zero, such that
$(x_n-x)\in \alpha_n B$ for every integer $n$.
\end{dfn}
One writes $x_n\Macto x$ to express the fact that the
sequence $x_n$ Mackey-converges to $x$.
Note that we could equivalently define Mackey convergence
in terms of a bounded subset $B$ which is not disked,
because the disked hull of a bounded set is bounded by
definition of convex bornological spaces.


A convex bornological space is called \emph{separated}
if the only vector subspace of $\calB$ is $\{0\}$.
A convex bornological space is separated iff every Mackey-convergent
sequence has a unique limit~\cite[p.~28]{Hogbe-77}.

\subsection{Completeness of $\calD'_\Gamma$}

The set of bounded maps from a convex bornological space
$E$ to $\bbK$ is called the \emph{bornological dual}
of $E$ and is denoted by $E^\times$.

A powerful theorem of bornology states~\cite[p.~77]{Hogbe-77}.
\begin{thm}
\label{completethm}
If a convex bornological space $E$ 
is regular (i.e. if  $E^\times$ separates points in 
$E$~\cite[p.~66]{Hogbe-77}), then
its bornological dual $E^\times$, endowed with its
natural topology, is a complete locally convex 
space.
\end{thm}

We are now going to build a bornological space $E$
such that $E^\times$ with its natural topology
is equal to $\calD'_\Gamma$ with its normal
topology. This implies the completeness of $\calD'_\Gamma$.

Recall that $E_\ell$ is the space
$\calE'_{\Lambda_\ell}(L_\ell)$ of the 
distributions compactly supported on $L_\ell$  whose
wavefront set is included in $\Lambda_\ell$,
where the family $(L_\ell)$ exhausts $\Omega$
and the family $(\Lambda_\ell)$ exhausts $\Lambda$.
To every locally convex space $E_\ell$ we associate
the convex bornological space ${}^bE_\ell$ which is
the vector space $E_\ell$ equipped with its
von Neumann bornology (i.e. the bornology defined by 
the bounded sets of the locally convex space 
$E_\ell$)~\cite[p.~48]{Hogbe-77}.
Let $E$ be the bornological inductive limit
of ${}^bE_\ell$, which is the vector space
$\calE'_\Lambda$ equipped with the bornology
defined by the bounded sets of $E_\ell$ for all
integers $\ell$~\cite[p.~33]{Hogbe-77}.

The bornological dual $E^\times$
of a convex bornological space $E$ is a locally convex
space for the natural topology defined by
the bounded sets of $E$. 
In other words, the seminorms of 
$E^\times$ are of the form
$p_{B'}(u)=\sup_{v\in B'}|\langle u,v\rangle|$,
where $B'$ runs over the bounded sets of $E$. 

We start by three lemmas, undoubtedly well-known to experts :
\begin{lem}\label{Mactoplem}
If $E$ is a quasi-complete, Hausdorff locally convex space
whose strong dual is a Schwartz space, then the 
Mackey-convergence of a sequence in $E$
is equivalent to its topological convergence.
In particular, this is the case for
$\calD'(\Omega)$ and $\calD'_\Gamma$.
\end{lem}
\begin{proof}
In a locally convex space, every Mackey-convergent
sequence (for the von Neumann bornology) is topologically
convergent~\cite[p.~26]{Hogbe-77}.
We have to prove that, conversely, any topologically convergent
sequence is also Mackey convergent.
Grothendieck~\cite{Grothendieck-54} showed that this holds
if the \emph{strict Mackey convergence
condition} is satisfied:
In a Hausdorff topological vector
space $E$, the strict Mackey convergence condition holds if,
for every compact subset $K$ of $E$, there is a bounded disk $B$
in $E$ such that $K$ is compact in $E_B=Span(B)$ (normed with 
the gauge of $B$, see~\cite[p.~26]{Hogbe-77},
\cite[p.~158]{Perez-Carreras},\cite[p.~285]{Horvath},
\cite{Gilsdorf-92}).

To show that this condition is satisfied with the hypotheses
of the lemma, we use the following theorem due
to Randtke~\cite{Randtke-72}:
\emph{Let $E$ be a locally convex Hausdorff space whose strong dual is a
Schwartz space. Then, for each precompact set $A$ of
$E$, there is a balanced, convex, bounded subset
$C$ of $E$ such that $C$ absorbs $A$ and
$A$ is a precompact subset of $E_C$.}
Thus, there is an $\alpha>0$ such that
$A\subset \alpha C$ and, if we denote
$\alpha C$ by $B$, we have a balanced, convex and bounded
subset $B$ of $E$ such that $A\subset B$ and $A$
is precompact in $E_B=E_C$. 

Consider a compact set $K$ in a locally convex space
$E$ that satisfies the hypotheses of the lemma.
According to Randtke's theorem, there is a balanced
convex and bounded subset $B$ containing $K$
for which $K$ is precompact in $E_B$.
The closure $\bar{B}$ of $B$ is a balanced, convex,
bounded and closed subset of $E$ 
such that the injection
$E_B\hookrightarrow E_{\bar{B}}$ is 
continuous~\cite[p.~II.26]{Bourbaki-TVS}.
Moreover, $K$ is also precompact in $E_{\bar{B}}$.
Indeed, $K$ is precompact in $E_B$ iff it is totally bounded, i.e. for
every neighborhood $V$ of zero, equivalently $V=\epsilon B, \epsilon>0$ , there is a finite
number of points $(x_i)_{1\le i \le m}$ of $E_B$ 
such that $K \subset 
\cup_{i=1}^m (x_i+V)$~\cite[p.~145]{Horvath}.
Since $E_B\subset E_{\bar{B}}$,
the points $x_i$ also belong to $E_{\bar{B}}$ and $K\subset 
\cup_{i=1}^m (x_i+\epsilon\bar{B})$
is precompact in $E_{\bar{B}}$.
The closed bounded set $\bar{B}$ is complete because
$E$ is a quasi-complete Hausdorff locally convex 
space~\cite[p.~128]{Horvath}.  As a consequence,
$E_{\bar{B}}$ is complete~\cite[p.~207]{Horvath}
and $K$ is compact in $E_{\bar{B}}$ because
every precompact set is relatively compact in a
complete space~\cite[p.~235]{Horvath}
and $K$ is closed in $E_{\bar{B}}$
(it is the inverse image of $K$ under the continuous
injection $E_{\bar{B}}\hookrightarrow E$~\cite[p.~97]{Schaefer}).
Therefore, $E$ satisfies the strict Mackey convergence
condition and the first part of the lemma is proved.

It remains to show that the conditions of the lemma
are fulfilled for $\calD'(\Omega)$ and $\calD'_\Gamma$.
We know that $\calD'(\Omega)$ is quasi-complete for 
the weak topology and complete for the strong topology.
Its strong dual is $\calD(\Omega)$, which is  
a Schwartz space~\cite[p.~282]{Horvath}. 
Therefore, the Mackey and topological sequential convergence
coincide in $\calD'(\Omega)$ with the weak and strong
topologies.

We proved that $\calD'_\Gamma$ is quasi-complete
with the H\"ormander topology (prop.~\ref{quasicompleteprop})
and is complete with the normal topology 
(cor.~\ref{completecor}). Its strong dual
$\calE'_\Lambda$ is a Schwartz space because
it is nuclear~\cite[p.~581]{KrieglMichor}.
Therefore, the Mackey and topological sequential convergence
coincide in $\calD'_\Gamma$ with the H\"ormander and
normal topologies.
\end{proof}

\begin{lem}\label{MackeyDense}
$\calD(\Omega)$ is Mackey-sequentially-dense in $E$.
\end{lem}
\begin{proof}
Take $u\in{}^bE_\ell=\calE'_{\Lambda_\ell}(L_\ell)$. It suffices to find 
$u_n\in \calD(\Omega)$ such that $u_n-u$ tends bornologically to 0 in 
$\calE'_{\Lambda_{\ell+1}}(L_{\ell+1})$.

From the proof of H\"ormander's density 
Theorem~\cite[p.~262]{HormanderI} we see that there exists 
a sequence $u_n \in \calD(\Omega)$ with $\supp(u_n)\subset L_{\ell+1}$ 
such that $u_n\to u$ in $\calD'_{\Lambda_{\ell+1}}$ and thus in 
$\calE'_{\Lambda_{\ell+1}}(L_{\ell+1})=E_{\ell+1}$.
The (topological) convergence of $u_n$ in $E_{\ell+1}$
implies its convergence in $\calD'(\Omega)$ and, by
lemma~\ref{Mactoplem}, its bornological convergence
in $\calD'(\Omega)$.
Thus, there exists a sequence $\alpha_n$ of positive
real numbers tending to zero and a disked bounded
set $B$ in $\calD'(\Omega)$ such that
$(u_n-u)\in \alpha_n B$ for every integer $n$.

However, we only know that $B$ is bounded in $\calD'(\Omega)$,
while we need to find a set which is bounded
in $E_{\ell+1}$ to show that $u$ is the bornological
limit of a sequence of test functions in $E_{\ell+1}$.
In other words, we still have to show that
$B$ is bounded for the
additional seminorms $||\cdot||_{N,V,\chi}$.

We already used in the proof of corollary \ref{PLS} that these additional
seminorms could be chosen in a
countable set $\{p_n\telque n\in \bbN^*\}$.
We can extract a subsequence $v_n$ from $u_n$ such that,
for all $k\le n$, $p_k(v_n-u)\le 1/n$. 
Hence, for every seminorm $p_k$, we have
$p_k(v_n-u)\le M_k/n$ for all positive integers $n$, where
$M_k=\sup_{n<k}\{n p_k(v_n-u),1\}$ is finite.
If we define the sequence $\beta_n=\max(\alpha_{N_n},1/n)$
of positive real numbers tending to zero,
the Mackey convergence of $u_n$ in $\calD'(\Omega)$ implies that,
for every integer $n$, there is an element $b_n$ of $B$
such that
$v_n-u=\alpha_{N_n}b_n=\beta_n(\alpha_{N_n}/\beta_n) b_n
=\beta_n c_n$ where $c_n=(\alpha_{N_n}/\beta_n) b_n\in B$
because $\alpha_{N_n}/\beta_n\le 1$ and $B$ is balanced.
Moreover, $p_k(v_n-u)/\beta_n\le 1/(n\beta_n) M_k \le M_k$.
Thus, for every $n$, $(v_n-u)/\beta_n$ belongs
to the set
$C=\{x\in B\cap E_{\ell+1}\telque p_k(x)\le M_k 
\text{ for every integer }k\}$,
which is balanced and bounded in $E_{\ell+1}$.

Finally, we have showed that
any distribution $u\in E_\ell$ is the
Mackey-limit in $E_{\ell+1}$ of a sequence of elements
of $\calD(\Omega)$ and the lemma is proved.
\end{proof}

\begin{lem}\label{UniformBound}
Let $B$ be a bounded set in $\mathcal{D}'(\Omega)$, then for every 
$f\in \mathcal{D}(\Omega)$ there exists $M$ such that 
$$\sup_{u\in B}\sup_{\xi\in \bbR^n}(1+|\xi|)^{-M}|\widehat{fu}(\xi)|<\infty.$$
\end{lem}

\begin{proof}This lemma is an obvious consequence of uniform boundedness principle.
Consider $(T_u)_{u\in B}$ the family of maps 
$T_u: C^\infty(\Omega)\to \mathbb{C}$ defined on the Fr\'echet space 
$C^\infty(\Omega) $ by $T_u(g)=u(fg).$ Since 
$fg\in \mathcal{D}(\Omega)$ and $B$ is weakly bounded, 
$\forall g\in C^\infty(\Omega),\exists C_{g}<\infty,\forall u\in
B$, $|T_u(g)|\leq C_{g}.$ Thus by the uniform boundedness principle, 
there exists 
a seminorm $p_l$ of $C^\infty(\Omega)$ such that 
$$\sup_{u\in B}|T_u(g)|\leq C p_l(g).$$

Since $\forall\xi\in \bbR^n, p_l(e_{\xi})\leq c(1+|\xi|)^M$ for some 
constants $c$ and $M$, this concludes.
\end{proof}

%

\begin{prop}\label{borndual}If $E$ is the bornological inductive limit of 
the spaces ${}^bE_\ell$ as above, then
$E^\times=\calD'_\Gamma$ and its natural topology is 
equivalent to the normal topology of $\calD'_\Gamma.$
\end{prop}

\begin{proof}
From lemma \ref{pairinglem} and proposition \ref{inductprop}, any 
$u\in\calD'_\Gamma$ defines a continuous linear form on each $E_\ell$ 
and thus a bounded linear form of ${}^bE_\ell$, 
i.e. an element of $(E)^\times$. 
This gives an embedding $\calD'_\Gamma\hookrightarrow(E)^\times$ since 
injectivity comes from the fact $\calD(\Omega)\subset E.$

Conversely, we want to prove that each bounded linear form $\lambda$ on $E$: (i) defines a distribution when restricted to $\calD(\Omega)\subset E$; 
(ii) has a wavefront set contained in $\Gamma$.

This will be enough to conclude the computation of the bornological dual 
since, from lemma \ref{MackeyDense} and the fact that  a bounded 
linear functional is Mackey-continuous~\cite[p.~10]{Hogbe-71},
the restriction of a bounded linear 
functional to  $\calD(\Omega)$ has a unique extension to $E$, proving 
that the second map above is injective.

To prove that $\lambda$ restricts to a distribution, we notice that
the injection
$\calD(L_\ell)\hookrightarrow E_\ell$ is continuous because
$E_\ell$ is a normal space of distributions.
Any bounded set $B$ of $\calD(\Omega)$, which is actually in some 
$E_\ell$, is
bounded in $E_\ell$ thus in $E$ because
the image of a bounded set by a continuous linear map
is a bounded set~\cite[p.~109]{Horvath}.
Thus, $\lambda$ is also a bounded map from 
$\calD(\Omega)$ to $\bbK$. It is well-known
that $\calD(\Omega)$ is bornological~\cite[p.~222]{Horvath}.
Hence, $\lambda$ is a continuous map from
$\calD(\Omega)$ to $\bbK$ because any bounded map from 
a bornological locally convex space to
$\bbK$ is continuous~\cite[p.~220]{Horvath}.
In other words, $\lambda$ is an element of $\calD'(\Omega)$.

We still have to show that $\lambda\in \calD'_\Gamma$, i.e. that
for any $\chi\in \calD(\Omega)$ and any closed convex 
neighborhood $V$ such that $\supp\chi\times V\cap \Gamma=\emptyset$,
the seminorm $||\lambda||_{N,V,\chi}$ is finite for all integers $N$.
For this we use again the remark made in the proof of 
proposition~\ref{firstdualprop} that 
$||\lambda||_{N,V,\chi}=\sup_{k\in V} |\lambda(f_k)|$,
where $f_k=(1+|k|)^N \chi e_k$.
Thus, if $B'=\{f_k\telque k\in V\}$ is a bounded
set in $E$, then we know that
$p_{B'}(\lambda)=\sup_{k\in V} |\lambda(f_k)| < {+}\infty$ because 
the image of the bounded set $B'$ by the bounded map $\lambda$
is bounded. It remains to show that $B'$ is a bounded
set of some $E_\ell$. We proceed as in the proof of
lemma~\ref{stronglem}.

First, $\supp\chi$ is a compact subset of the
open set $\pi_1(\Lambda)$. Therefore, there is an integer
$\ell$ such that $L_\ell$ is a compact neighborhood
of $\supp\chi$ and
$U^*\Omega \cap \Lambda_\ell$ is a compact neighborhood of
$U^*\Omega\cap (\supp\chi\times (-V))$ because
$L_\ell$ exhausts $\Omega$ and $\Lambda_\ell$ exhausts $\Lambda$.
This space $E_\ell$ contains $B'$
because each $f_k$ is smooth and compactly supported and we want 
to show that $B'$ is bounded in this $E_\ell.$

Consider $||f_k||_{N',W,\psi}$ where
$\supp\psi\times W \cap \Lambda_\ell=\emptyset$.
If $\supp\psi\cap\supp\chi=\emptyset$, then
$||f_k||_{N',W,\psi}=0$.
If $\supp\psi\cap\supp\chi\not=\emptyset$, then $W\cap (-V)=\emptyset$ 
and thus, by compactness of the intersections of these cones with the 
unit sphere, there is a $c>0$ such that $|k+q|/|q|>c$ and
$|k+q|/|k|>c$ for all $k\in V,q\in W$. 
We follow the proof of proposition~\ref{firstdualprop}
to show that 
\begin{eqnarray*}
||f_k||_{N',W,\psi} & \le & 
 c^{-N-N'}
 \sup_{q\in W}
 (1+|k+q|)^{N+N'}
|\widehat{\psi\chi}(k+q)|.
\end{eqnarray*}

According to eq.~\eqref{ekNVchi}, there is a constant
$C_{N+N',\psi\chi}$ such that
$||f_k||_{N',W,\psi}  \le  
 c^{-N-N'} C_{N+N',\psi\chi}$.
Therefore, 
$||f_k||_{N',W,\psi}$
is uniformly bounded for all values of $k\in V$.

To conclude the proof of the boundedness of $B'$ in $E_\ell$, we show that $p_B(f_k)$ is bounded for all
bounded sets $B\subset\calD(\Omega)$. We know that
$\calD(\Omega)$ is a Montel space~\cite[p.~357]{Treves}.
Thus, it is barrelled and it is enough to show that
$B'$ is weakly bounded: i.e. that, for any $g\in \calD(\Omega)$,
$\langle f_k,g\rangle$ is bounded. Indeed we have
$|\langle f_k,g\rangle| =
(1+|k|)^N |\langle e_k,\chi g\rangle| = 
(1+|k|)^N |\widehat{\chi g}(k)|$,
which is bounded uniformly in $k\in \bbR^n$,
as seen from eq.~\eqref{ekNVchi}.

Finally, we have shown that $B'$ is bounded in $E_\ell$,
which implies that $B'$ is bounded in $E$
and that $||\lambda||_{N,V,\chi}=p_{B'}(\lambda)< {+}\infty$
for all integers $N$ and all $V,\chi$ such that
$\supp\chi\times V\cap\Gamma=\emptyset$. 
This concludes our proof of $\WF(\lambda) \subset \Gamma$.

Moreover, this also shows that the natural  topology 
of $E^\times$ is finer 
than the normal topology of $\calD'_\Gamma$. 
Indeed,
we proved that, for any seminorm $||\cdot||_{N,V,\chi}$
of $\calD'_\Gamma$, there is a bounded set $B'$ in 
$E$ such that $||\cdot||_{N,V,\chi}=p_{B'}$ and
the seminorms $p_B$, where $B$ is bounded in $\calD(\Omega)$
are both in $\calD'_\Gamma$ and $E^\times$.
In other words, $E^\times$ has more seminorms than $\calD'_\Gamma$.

It remains to show the converse, i.e. the continuity of the 
injection $\calD'_\Gamma\mapsto E^\times$.  
For this we have to describe more precisely $E^\times$,
which is the bornological dual of a bornological inductive limit.
In the topological case, it is well known that the 
dual of an inductive limit is a projective
limit~\cite[p.~85]{Robertson}\cite[p.~290]{Kothe-I}. 
We have a similar result for the bornological case.
Indeed, $^bE_\ell$ is the vector space
$E_\ell=\calE'_{\Lambda_\ell}(L_\ell)$ equipped with
the bornology $\calB_\ell$ whose elements are the
subsets of $E_\ell$ which are bounded in $\calD'_{\Lambda_\ell}$.
The injection $j_\ell:{}^bE_\ell\to {}^bE_{\ell+1}$
is bounded because it is continuous. 
Thus, $E=\cup_{\ell} E_\ell$ is a convex bornological
space whose bornology is 
$\calB=\cup_\ell \calB_\ell$~\cite[p.~33]{Hogbe-77}
\cite[p.~195]{Wong}.

By duality, $E^\times=\cap_\ell ({}^bE_\ell)^\times$
as a vector space. The (algebraic) dual map
$j_\ell^*: {}^bE_{\ell+1}^\times \to {}^bE_\ell^\times$
is the inclusion. It is continuous if every
$({}^bE_\ell)^\times$ is equipped with 
its natural topology (since
$\calB_\ell\subset\calB_{\ell+1}$, every seminorm
$p_B$ of ${}^bE_\ell^\times$ where $B\in \calB_\ell$ is also 
a seminorm of ${}^bE_{\ell+1}^\times$).
Therefore, $E^\times$ is the topological
projective limit of the 
locally convex spaces 
$({}^bE_\ell)^\times$~\cite[p.~230]{Kothe-I}.
To show that the injection
$\calD'_\Gamma \hookrightarrow E^\times$ 
is continuous, we just have to prove that each injection
$\calD'_\Gamma \hookrightarrow ({}^bE_\ell)^\times$ 
is continuous~\cite[p.~149]{Wong}.

Said otherwise, we have to show that the 
bound~(\ref{paireq}) we obtained in lemma~\ref{pairinglem} can be 
made uniform in $v\in B$ for some bounded set $B$ in $E_\ell.$ 
First note that the choices of functions $\psi,\alpha,\beta$ 
can be made uniformly for $v\in B$,  $B$ a bounded set in $E_\ell.$
Second, using lemma \ref{UniformBound}, one sees
that the constants $m,C$ used in 
the proof of the bound (\ref{paireq}) can be made uniform in $v\in B$ 
so that $\sup_{v\in B}|\widehat{v\psi_{j}}(k)|
\leq C(1+|k|)^m.$ 
Moreover, by definition of boundedness 
$\sup_{v\in B}||v||_{N,U_{\beta j},\psi_j} 
\leq M_{N,U_{\beta j},\psi_j}$.

We thus obtain:
\begin{eqnarray}
p_B(u)=\sup_{v\in B}|\langle u,v\rangle| & \le & \sum_j
\Big(
p_{B_j'}(u) + ||u||_{M,U_{\alpha j},\psi_j} C
I_n^{M-m} 
\nonumber\\&& + ||u||_{M,U_{\alpha j},\psi_j} M_{N,U_{\beta j},\psi_j} 
  I_n^{N+M} \Big),
\end{eqnarray}
where $B_j':=\{\psi_jf_j^v \telque v\in B \}$ with 
$\widehat{f^v_j}(k)=\alpha_j(-k)(1-\beta_j(k))
\widehat{\psi_j v}(k)$. To prove the expected continuity, it thus only 
remains to show that $B'$ is bounded in $\calD(\Omega)$ 
so that $p_{B'}$ is a seminorm of $\calD'_\Gamma$.

But, let $K_j=\supp \psi_j$,  we deduce:

\begin{align*}\pi_{N,K_j}(\psi_jf_j^v)&\leq 2^N\pi_{N,K_j}(\psi_j)
\pi_{N,K_j}(f_j^v) 
\\&\leq 2^N\pi_{N,K_j}(\psi_j)
(2\pi)^{-n}\sup_{|\gamma|\leq N}|\int_{\supp \beta_j'}dk
(k^\gamma)\alpha_j(-k)(1-\beta_j(k))
\widehat{\psi_j v}(k)|
\\&\leq 2^N\pi_{N,K_j}(\psi_j)I_{n}^{n+1}||v||_{N+n+1,U_{\beta
j},\psi_j},\end{align*}
and the last seminorm is a seminorm in $E_\ell$ since 
$U_{\beta j}=\supp(1-\beta_j)$ 
has been chosen (in the process of choosing $\psi,\alpha,\beta$)
independent of $v\in  B$, so that 
$\supp(\psi_j)\times \supp(1-\beta_j)\cap \Lambda_\ell=\emptyset.$ 
The above estimate thus concludes.
\end{proof}

\begin{cor}
\label{completecor}
$\calD'_\Gamma$ with its normal topology is complete.
\end{cor}
\begin{proof}
From theorem~\ref{completethm}, it remains to check 
that $E$, as a convex bornological space, 
is  regular.  From our computation of the dual, 
it was already proved in lemma~\ref{pairinglem} 
that $E^\times$ separates points in $E$.
Thus, $E$ is a regular convex bornological space and its dual
$\calD'_\Gamma$ is complete with its normal topology, because it
is equivalent to the natural topology.
\end{proof}

\subsection{$\calE'_\Lambda$ is ultrabornological}
\label{Ebornosect}
A locally convex space is \emph{bornological} if its balanced,
convex and bornivorous subsets are neighborhoods of 
zero~\cite[p.~220]{Horvath}.
Bornological spaces have very convenient properties. For example,
every linear map $f$ from a bornological locally convex space $E$
to a locally convex space $F$ is continuous iff
it is bounded (i.e. if $f$ sends every bounded set of $E$
to a bounded set of $F$)~\cite[p.~220]{Horvath}.

\begin{prop}
\label{ELbornoprop}
$\calE'_\Lambda$ is a bornological locally convex space.
\end{prop}
\begin{proof}
By a standard theorem~\cite[p.~221]{Horvath},
a locally convex Hausdorff space $E$ is bornological iff
the topology of $E$ is the Mackey topology and any bounded
linear map from $E$ to $\bbK$ is continuous.
We already know from lemma~\ref{stronglem} that the inductive topology
on $\calE'_\Lambda$ is equivalent to the Mackey topology.
Thus, it remains to show that
a linear map $\lambda:\calE'_\Lambda\to\bbK$ is continuous if
$\sup_{v\in B'} |\lambda(v)| < \infty$ for every
bounded subset $B'$ of $\calE'_\Lambda$.
Since $\lambda$ is a fortiori bounded for the coarser bornology of $E$,
we know from proposition \ref{borndual} that it defines by restriction
on $\calD(\Omega)$ an element of $\calD'_\Gamma$. Then this element
extends to a continuous linear form on $\calE'_\Lambda$ and since, by
lemma~\ref{MackeyDense}, $\calD(\Omega)$ is Mackey dense in $E$ 
and a fortiori in $\calE'_\Lambda$, the extension has to 
coincide with the
original $\lambda$ (which is bounded thus Mackey sequentially
continuous). Therefore, $\lambda$ is continuous.
\end{proof}

Note that the previous argument says $\calE'_\Lambda$ has the same 
bornological dual as $E$, but not necessarily with the same natural topology. 
Indeed, the natural topology of
$(\calE'_\Lambda)^\times$ is the strong
$\beta(\calD'_\Gamma,\calE'_\Lambda)$ topology on $\calD'_\Lambda$.
If the normal topology of $\calD'_\Lambda$ were the
strong topology, then $\calE'_\Lambda$ would be semi-reflexive
because the dual of $\calD'_\Gamma$ for the normal topology is
$\calE'_\Lambda$. Thus, $\calE'_\Lambda$ would be quasi-complete
and we shall prove in section~\ref{comEsect} that this is not
the case when the open cone $\Lambda$ is not closed.

This implies another consequence regarding the regularity
of the inductive limit. 
Recall that an inductive limit of locally convex spaces is said
to be \emph{regular} if each bounded set of $E$
is contained and bounded in some 
$E_\ell$~\cite{Kucera-78,Qiu-00}.
If the inductive limit defining the topology of 
$\calE'_\Lambda$ were regular, then the bornology
of $\calE'_\Lambda$ would be the bornology of $E$
(because we already know that every bounded set of
$E$ is bounded in $\calE'_\Lambda$). In that case,
the natural topologies of their bornological dual
$\calD'_\Gamma$ would be identical and
the normal topology on $\calD'_\Gamma$ would be the strong topology.
Thus, the inductive limit is not regular when
$\Lambda$ is not both open and closed.

Let us see how this bornological property also follows from a general theorem, even giving us a
stronger result:
\begin{prop}
$\calE'_\Lambda$ is an ultrabornological locally convex space.
\end{prop}
\begin{proof}
$\calD'_\Gamma$ is complete and nuclear.
Therefore, by noticing that 
any nuclear locally convex space is 
Schwartz~\cite[p.~581]{KrieglMichor},
we see that
$\calE'_\Lambda$ is ultrabornological
because it is the strong
dual of a complete Schwartz locally convex 
space~\cite[p.~287]{Horvath} \cite[p.~15]{Hogbe-81}.
\end{proof}
Note that ultrabornological spaces are also
called completely bornological~\cite[p.~53]{Hogbe-77}
or fast-bornological~\cite[p.~203]{Wong}.
A locally convex space is ultrabornological
iff it is the topologification of a complete
convex bornological space~\cite[p.~53]{Hogbe-77}.
An ultrabornological space is the inductive
limit of a family of separable Banach spaces~\cite[p.~274]{Jarchow}.
Further characterizations are
known~\cite{Valdivia-74},\cite[p.~207-210]{Hogbe-81},
\cite[Ch.~6]{Perez-Carreras},
\cite[p.~283]{Meise}, \cite[p.~54]{Gach-04}.
The relation between boundedness and continuity is:
A linear map from an ultrabornological space $E$ to
a locally convex space $F$ is continuous iff it
is bounded on each compact disk of $E$~\cite[p.~54]{Hogbe-77}.

\section{Functional properties of $\calD'_\Gamma$ and $\calE'_\Lambda$}
In this section, we put together the results derived up to now
to determine the main functional properties
of $\calD'_\Gamma$ and $\calE'_\Lambda$.

\subsection{General functional properties}

\begin{prop}\label{general}
The space $\calD'_\Gamma$ is a normal space of distributions.
It is Hausdorff, nuclear and semi-reflexive.
Its topological dual is $\calE'_\Lambda$ which 
is Hausdorff, nuclear, and barrelled.
\end{prop}
\begin{proof}
We saw that $\calD'_\Gamma$ is Hausdorff. Its dual
$\calE'_\Lambda$ is also Hausdorff because the pairing
$\langle \cdot,\cdot\rangle$ is separating
(see lemma~\ref{pairinglem}) and the topology of
$\calE'_\Lambda$ is finer than the weak
topology $\sigma(\calE'_\Lambda,
\calD'_\Gamma)$~\cite[p.~185]{Horvath}.
We proved that $\calD'_\Gamma$ is the dual
of $\calE'_\Lambda$ for the inductive 
topology and that the inductive topology
of $\calE'_\Lambda$ is equivalent to the strong topology
$\beta(\calE'_\Lambda, \calD'_\Gamma)$.
Therefore, $\calD'_\Gamma$ is the topological 
dual of $\calE'_\Lambda$, which is the
strong dual of $\calD'_\Gamma$.
This implies that $\calD'_\Gamma$ is 
semi-reflexive~\cite[p.~227]{Horvath}.

The space $\calE'_\Lambda$ is barrelled because it
is the strong dual of a semi-reflexive 
space~\cite[p.~228]{Horvath}.
This can also be deduced from the fact that
the inductive topology of $\calE'_\Lambda$
is equal to its strong topology~\cite[p.~IV.5]{Bourbaki-TVS}.
\end{proof}

In fact, $\calD'_\Gamma$ is even 
a completely reflexive locally convex space, because it is 
complete and Schwartz~\cite[p.~95]{Hogbe-77}.
Recall that a locally convex space $E$
is completely reflexive 
(or ultra-semi-reflexive~\cite[p.~243]{Wong})
if $E=(E')^\times$ algebraically
and topologically~\cite[p.~89]{Hogbe-77},
where $E'$ is the dual of $E$ with the equicontinuous
bornology and $(E')^\times$ is the bornological
dual of $E'$ with its natural topology.
This has two useful consequences: (i)
$\calE'_\Lambda$ equipped with
the equicontinuous bornology is a reflexive
convex bornological space~\cite[p.~136]{Hogbe-77};
(ii) the strong and ultra-strong topologies on
$\calE'_\Lambda$ are equivalent~\cite[p.~90]{Hogbe-77}.

\subsection{Completeness properties of $\calD'_\Gamma$}
We state the results concerning the completeness
of $\calD'_\Gamma$:
\begin{prop}
\label{quasicompleteprop}
In $\calD'_\Gamma$:
\begin{itemize}
\item $\calD'_\Gamma$ is complete for all topologies finer
  than the normal topology and coarser
    than the Mackey topology. 
\item $\calD'_\Gamma$ is quasi-complete for all topologies
  compatible with the duality between $\calD'_\Gamma$
  and $\calE'_\Lambda$:
   all the bounded closed subsets are complete for these
  topologies. In particular, $\calD'_\Gamma$ is quasi-complete
  for the H\"ormander topology.
\end{itemize}
\end{prop}
\begin{proof}
We have proved that $\calD'_\Gamma$ is complete for
the normal topology. Thus, it is complete for all
topologies that are finer than the normal topology
and that are compatible with duality~\cite[p.~IV.5]{Bourbaki-TVS}.
We have also showed that $\calD'_\Gamma$ is semi-reflexive.
As a consequence, it is quasi-complete for the weak topology 
$\sigma(\calD'_\Gamma,\calE'_\Lambda)$~\cite[p.~228]{Horvath}.
This implies that $\calD'_\Gamma$ is quasi-complete
for every topology compatible with the duality
between $\calD'_\Gamma$ and $\calE'_\Lambda$,
in particular for the normal topology~\cite[p.~IV.5]{Bourbaki-TVS}.
Since Bourbaki's proof is rather sketchy, we give it in more detail.
Assume that $E$ is quasi-complete for the weak topology
$\sigma(E,E')$ and consider a topology $\calT$ compatible
with duality. The space $E$ is quasi-complete for $\calT$
iff every $\calT$-closed $\calT$-bounded subset of $E$ is 
complete~\cite[p.~128]{Horvath}.
Consider a subset $C$ of $E$ which is closed and bounded for $\calT$.
By the theorem of the bipolars, the bipolar $C^{\circ\circ}$ of
$C$ is a balanced, convex, $\sigma(E,E')$-closed
set containing $C$. We also know that $C$ is bounded for $\calT$
iff it is bounded for $\sigma(E,E')$ because $\calT$
is compatible with duality~\cite[p.~209]{Horvath}. Then, we
use the fact that $C$ is bounded for $\sigma(E,E')$
iff $C^\circ$ is absorbing~\cite[p.~191]{Horvath}.
But $C^\circ=(C^{\circ\circ})^\circ$ so that
$C^{\circ\circ}$ is weakly bounded if and only if $C$ is weakly bounded.
Therefore, $C^{\circ\circ}$ is bounded, convex and closed for
$\sigma(E,E')$, and also for the other topologies compatible with
duality by the first two items of the proposition.
Consider now a Cauchy filter on $C^{\circ\circ}$ for the topology $\calT$.
It is also a Cauchy filter for the weak topology.
Indeed a filter $\frakF$ is Cauchy if and only if, for any neighborhood
$V$ of zero, there is an $F\in\frakF$ such that
$F-F\subset V$. The topology $\calT$ being compatible with
duality, it is finer than the weak topology. Thus,
any weak neighborhood $V$ is also a neighborhood of $\calT$
and $\frakF$ is a Cauchy filter for the weak topology.
This Cauchy filter converges to 
a point $x$ because $E$ is quasi-complete for the weak topology.
Moreover, $x$ is in $C^{\circ\circ}$ because $C^{\circ\circ}$ is 
weakly closed. Therefore, the Cauchy filter converges in $C^{\circ\circ}$
and $C^{\circ\circ}$ is complete for $\calT$. As a consequence,
$C$ itself is also complete because it is a closed 
subset of a complete set~\cite[p.~128]{Horvath}.
\end{proof}

This brings us to the following result
\begin{prop}
\label{semiMontel}
The space $\calD'_\Gamma$ with its normal topology is semi-Montel.
The space $\calE'_\Lambda$ is a normal space of distributions
on which the strong, Mackey, inductive limit
and Arens topologies are equivalent.
\end{prop}
\begin{proof}
We saw that $\calD'_\Gamma$ is quasi-complete and nuclear
for its normal topology. Thus, its bounded subsets are
relatively compact~\cite[p.~520]{Treves} and 
$\calD'_\Gamma$ is semi-Montel by definition of
semi-Montel spaces~\cite[p.~231]{Horvath}.
We already know that the strong, Mackey and inductive limit
topologies are equivalent. It is known that on the dual of a
semi-Montel space, the Arens topology is equivalent to
the strong and Mackey ones~\cite[p.~235]{Horvath}.
By item~(iii) of proposition~\ref{normalprop},
we obtain that $\calE'_\Lambda$ is a normal space of distributions.
\end{proof}

Semi-Montel spaces have interesting stability 
properties~\cite[\S~3.9]{Horvath},
\cite[\S~11.5]{Jarchow}
(for example, a closed subspace of a semi-Montel 
space is semi-Montel~\cite[p.~232]{Horvath},
as well as a strict inductive limit of semi-Montel 
spaces~\cite[p.~240]{Horvath}).
Moreover, if $B$ is a bounded subset of $\calD'_\Gamma$,
then the topology induced on $B$ by the normal
topology is the same as that induced by the weak
$\sigma(\calD'_\Gamma,\calE'_\Lambda)$ 
topology~\cite[p.~231]{Horvath} and
$B$ is metrizable (because
$\calE'_\Lambda$, the strong dual of $\calD'_\Gamma$,
is nuclear~\cite[p.~217]{Hogbe-81}).

The following properties of semi-Montel spaces are a
characterization of convergence~\cite[p.~232]{Horvath}
which is useful in renormalization theory:
\begin{prop}
If $u_i$ is a sequence of elements of $\calD'_\Gamma$ such that 
$\langle u_i,v\rangle$ converges to some number
$\lambda(v)$ in $\bbK$ for all $v\in \calE'_\Lambda$, then 
the map $u:v\mapsto \lambda(v)$ belongs to $\calD'_\Gamma$ and
$u_i$ converges to $u$ in $\calD'_\Gamma$.
\end{prop}
\begin{prop}
If $(u_\epsilon)_{0<\epsilon < \alpha}$ is a family of elements of 
$\calD'_\Gamma$ such that 
$\langle u_\epsilon,v\rangle$ converges to some number
$\lambda(v)$ in $\bbK$ as $\epsilon\to 0$
for all $v\in \calE'_\Lambda$, then 
the map $u:v\mapsto \lambda(v)$ belongs to $\calD'_\Gamma$ and
$u_\epsilon\to u$ in $\calD'_\Gamma$
as $\epsilon\to0$.
\end{prop}
By proposition~\ref{quasicompleteprop}, we see that $\calD'_\Gamma$
is quasi-complete for the H\"ormander topology. However, it
is generally not complete because $\calD'(\Omega)$ is not
complete for the weak topology (otherwise, every linear map
from $\calD(\Omega)$ to $\bbK$ would be continuous, whereas 
it is well known that the algebraic dual of $\calD(\Omega)$
is larger than $\calD'(\Omega)$~\cite{Oberguggenberger-13}).

%

\subsection{Bounded sets}
The bounded sets of $\calD'_\Gamma$ are important in
renormalization theory because they are used to
define the scaling degree~\cite{Brunetti2} of
a distribution and the weakly homogeneous 
distributions~\cite{Meyer-98}.

The bounded sets of $\calD'_\Gamma$ were characterized 
in the proof of lemma~\ref{stronglem}: a subset $B'$
of $\calD'_\Gamma$ is bounded if 
$B'$ is a bounded set of $\calD'(\Omega)$
and for every integer $N$, every $\psi\in \calD(\Omega)$ and
every closed cone $V$ such that
$\supp\psi\times V \cap \Gamma=\emptyset$, there
is a constant $M_{N,V,\chi}$ such that
$||u||_{N,V\chi} \le M_{N,V,\chi}$ for all $u\in B'$.
The bounded sets of $\calD'(\Omega)$ have several 
characterizations (see \cite[pp.~86 and 195]{Schwartz-66}
and \cite[pp.~330 and 493]{Edwards}).

We can now give a list of the main properties of the
bounded sets of $\calD'_\Gamma$, which correspond
to a Banach-Steinhaus theorem for $\calD'_\Gamma$:
\begin{thm}
\label{boundsetprop}
In $\calD'_\Gamma$:
\begin{itemize}
\item The bounded subsets
are the same for all topologies
finer than the weak topology $\sigma(\calD'_\Gamma,\calE'_\Lambda)$
and coarser than the strong topology
$\beta(\calD'_\Gamma,\calE'_\Lambda)$.
In particular, they are the same
for the normal and the H\"ormander topologies.
\item The bounded sets are equicontinuous.
\item The closed bounded sets are compact, identical
and topologically equivalent
for the weak, H\"ormander and
normal topologies.
\end{itemize}
\end{thm}
\begin{proof}
In general, the bounded subsets of a topological vector space
$E$ are the same for all locally convex Hausdorff topologies on $E$
compatible with the duality between
$E$ and $E'$~\cite[p.~371]{Treves},
i.e. for all topologies finer than the weak topology
and coarser than the Mackey topology~\cite[p.~369]{Treves}.
The barrelledness of $\calE'_\Lambda$ implies that
these bounded sets are also identical with the strongly 
bounded sets~\cite[p.~212]{Horvath}.
In the dual $\calD'_\Gamma$ of the barrelled space $\calE'_\Lambda$, 
a set is bounded if and only if it is equicontinuous~\cite[p.~212]{Horvath}.
In a quasi-complete nuclear space,
every closed bounded subset is compact~\cite[p.~520]{Treves}.
Especially, using propositions \ref{general} and \ref{quasicompleteprop}, this implies that bounded subsets closed for the H\"ormander and normal topologies are compact for these topologies.
In the dual of a barrelled space, the weakly closed 
bounded sets are weakly compact~\cite[p.~212]{Horvath}.
After the proof of prop.~\ref{firstdualprop}, we showed
that the H\"ormander topology is compatible with the 
pairing~\cite[p.~198]{Horvath}.
Thus, by the Mackey-Arens theorem~\cite[p.~205]{Horvath},
it is finer than the weak topology and coarser
than the Mackey one. 

In the remarks following Proposition~\ref{semiMontel}, we 
showed that the weak and normal topologies are equivalent
on the bounded sets. Therefore, the H\"ormander topology
is equivalent to those since it is finer than the weak topology
and coarser than the normal one. As a consequence,
the closed and bounded sets are the same for the
three topologies. Indeed, it suffices to remember that the
bounded sets closed for one of these topologies are compact for the 
corresponding induced topology, and compactness is an internal topological 
property so that they are compact for all the induced topologies since they 
coincide. Finally, compactness implies in a Hausdorff space 
that they are closed for the three topologies.
\end{proof}

In concrete terms, this means that
a subset $B'$ is bounded in $\calD'_\Gamma$
if and only if one (and then all) of the following conditions is
satisfied:
\begin{itemize}
\item[(i)] For every $v\in \calE'_\Lambda$, there
   is a constant $M_v$ such that
   $|\langle u,v\rangle|\le M_v$
  for all $u\in B'$. This defines weakly bounded sets.
\item[(ii)] For every bounded set $B$ of $\calE'_\Lambda$,
  there is a constant $M_B$ such that
   $|\langle u,v\rangle|\le M_B$ for 
  all $u\in B'$ and all $v\in B$. This defines
  strongly bounded sets. 
\item[(iii)] There is a constant $C$ and a finite
  set of seminorms $p_i$ of $\calE'_\Lambda$
  such that 
  $|\langle u,v\rangle| \le C \max_i p_i(v)$.
  This defines equicontinuous sets~\cite[p.~200]{Horvath}.
\end{itemize}
With respect to item (ii) recall that, the inductive
limit being not regular, there are bounded
sets in $\calE'_\Lambda$ that are not contained
and bounded in any $E_\ell$.
However, of course, as we already used, the bounded sets of 
every $E_\ell$ are bounded in $\calE'_\Lambda$.

Note also that the closed \emph{convex} subsets are the same for all topologies
  compatible with the duality between $\calD'_\Gamma$
  and $\calE'_\Lambda$~\cite[p.~370]{Treves}.

\subsection{Completeness properties of $\calE'_\Lambda$}
\label{comEsect}
By contrast with $\calD'_\Gamma$, the completeness properties
of $\calE'_\Lambda$ are very poor. 
More precisely, we have
\begin{thm}
Assume that $\Lambda$ is an open cone which is not closed, 
then $\calE'_\Lambda$ 
with its strong topology is not 
(even weakly) sequentially complete. 
In particular, if $\Omega$ is connected and the dimension of spacetime
is $n>1$, then $\calE'_\Lambda$ is not sequentially complete 
when $\Lambda$ is any open conical nonempty proper subset of
$\dotT^*\Omega$.
\end{thm}
\begin{proof}
In fact, if $\Lambda$ is an open cone
which is not closed in $\dotT^*\Omega$, we exhibit an 
explicit counterexample showing that
$\calE'_\Lambda$ is not sequentially complete.
Since the construction of this counterexample is
a bit elaborate, we first describe its main ideas.
Consider a point $(x;\eta)$ in the boundary of $\Lambda$.
There is a sequence of points $(x_m;\eta_m)\in \Lambda$
such that $(x_m;\eta_m)\to (x;\eta)$.
By using an example due to H\"ormander, we construct
a distribution $v_m$ whose wavefront set is exactly
the line $\{(x_m;\lambda\eta_m)\telque \lambda>0\}$.
Then we show that the sum $v=\sum_m v_m/m!$ is 
a well-defined distribution which does not belong to
$\calE'_\Lambda$ because the point $(x;\eta)$
belongs to its wavefront set. Since the series
defining $v$ is a Cauchy sequence, we have defined
a Cauchy sequence in $\calE'_\Lambda$ 
whose limit is not in $\calE'_\Lambda$.

The proof consists of several steps:
(i) description of H\"ormander's example,
(ii) construction of the counter-example
  $v=\sum v_m/m!$,
(iii) choice of the sequence $(x_m;\eta_m)$
and of the closed cones $\Gamma_M$,
(iv) calculation of the seminorms of $v_m$
  in $\calD'_{\Gamma_M}$,
(v) determination of the wavefront set of $v$,
(vi) proof that the series is Cauchy in 
  $\calE'_\Lambda$,
(vii) discussion of the case where 
$\Lambda$ is both open and closed.

\begin{step}
H\"ormander's distribution
\end{step}

To build this counterexample we start
from a family of  distributions, defined by 
H\"ormander~\cite[p.~188]{Hormander-97}, whose
wavefront sets are made of a single point $x$
and a single direction $\lambda k$ and whose
order is arbitrary:
Let $\chi\in C^\infty(\bbR,[0,1])$ be equal to 1
in $(-\infty,1/2)$ and to 0 in $(1,+\infty)$, 
with $0\le \chi\le 1$. Fix $0 < \rho < 1$,
let $\eta\in \mathbb{R}^n$ be a unit vector, 
take an orthonormal basis 
$(e_1=\eta,e_2,...,e_n)$ and write coordinates in this coordinate system.

Define $u_{\eta,s}\in \calS'(\bbR^n)$, for $s\in \bbR$, by
\begin{eqnarray*}
\widehat{u_{\eta,s}}(\xi) &=& (1-\chi(\xi_1)) \xi_1^{-s}
  \chi((\xi_2^2+\dots+\xi_n^2)/\xi_1^{2\rho}).
\end{eqnarray*}
 Then
$\WF(u_{\eta,s})=\{(0;\xi); \xi_2=\dots=\xi_n=0, \xi_1>0\}=\{0\}\times \bbR_+^*\eta$
and $u_{\eta,s}$ coincides with a function in $\calS(\bbR^n)$
outside a neighborhood of the origin~\cite[p.~188]{Hormander-97}.
It is clear that, if
$\xi=\lambda \eta$ and $\lambda>1$,
then
$|\widehat{u}_{\eta,s}(\xi)|=\lambda^{-s}$ for
any $\lambda>1$, where $s$ is an arbitrary
real number. Thus, the degree of growth
can be an arbitrary polynomial degree.
Moreover, H\"ormander actually proves that 
for any real number $t$ and
any integer $m$, there is a constant $C(t,m)$, such that if 
$\alpha,\beta$ are multi-indexes, and $|\alpha|\geq C(t,m)$ 
then 
$\{x^{\alpha}\partial^\beta u_{\eta,s}, s\geq t, 
|\beta|\leq m,\eta\in S^{n-1}\}$
are bounded continuous functions  on $\bbR^n$, 
uniformly bounded by a constant $D(t,m).$

One should also note that when the last factor in the definition does 
not vanish, we have $\xi_1^{2\rho}\geq (\xi_2^2+...+\xi_n^2)$ so that 
$|\xi_1|^2\geq \frac{|\xi_1|^2+ (\xi_2^2+...+\xi_n^2)^{1/\rho}}{2}\geq 
|\xi|^2/2$ as soon as $\xi_2^2+...+\xi_n^2\geq 1,$ and otherwise 
$|\xi|^{2}\leq |\xi_1|^2+1\leq 5|\xi_1|^2$ when $(1-\chi)(\xi_1)\neq 0$ 
(which implies $\xi_1\geq 1/2$). Moreover, when the first factor does not 
vanish $|\xi|\geq 1/2$ so that $|\xi|\geq (1+2|\xi|)/4\geq (1+|\xi|)/4$. 
As a consequence, we note for $s\geq 0$: 
\begin{align}\label{TechIneq2}|\widehat{u_{\eta,s}}(\xi)| 
\leq (1-\chi(\xi_1))80^{s/2} (1+|\xi|)^{-s}
\chi((\xi_2^2+\dots+\xi_n^2)/\xi_1^{2\rho})
\le 10^s (1+|\xi|)^{-s}.\end{align}

\begin{step}
Construction of the counterexample
\label{stepconst}
\end{step}
Since $\Lambda$ is open and not closed, its boundary
$\partial\Lambda=\overline{\Lambda}\backslash\Lambda$ is not
empty and $\partial\Lambda\cap\Lambda=\emptyset$~\cite[p.~46]{Kelley}.
Moreover, any point $(x;\eta)$ of $\partial\Lambda$ is the limit of a sequence
of points $(x_m;\eta_m)$ in $\Lambda$~\cite[p.~9]{Bredon}.

By starting from H\"ormander's example, we build
a family of distributions $v_m$ such that the wavefront
set of $v_m$ is $\{(x_m;\lambda \eta_m)\telque\lambda >0\}$ and 
$|\widehat{v_m}(\lambda \eta_m)|=(\lambda|\eta_m|)^{-m}$.
For this we use the translation operator $T_x$ acting on
test functions by $(T_x f)(y)=f(y-x)$ and extend it to
distributions by $\langle T_x u,f\rangle = \langle u, T_{-x}f\rangle$.
Thus $T_{x_m}u_{\eta_m,m}$ has the desired properties.
However, we want all distributions $v_m$ to be compactly supported
on $\Omega$. Thus, we define the compact set 
$X=\cup_{m=1}^\infty\{x_m\}\cup \{x\}\subset \Omega$, so that
$\delta=d(X,\Omega^c)>0$, and $\chi$ a smooth function
compactly supported on $B(0,\delta/2)$ and equal to 1
on a neighborhood of the origin.
Then $v_m=T_{x_m} (\chi u_{\eta_m,m})$ is a 
distribution in $\calE'(\Omega)$ with the desired properties.

It is easy to show that the series $v=\sum_{m=1}^\infty v_m/m!$
converges to a distribution in $\calE'(\Omega)$.
Indeed, it is enough to  prove that, for any $f\in \calD(\Omega)$,
the numerical series $\sum_m \langle v_m,f\rangle /m!$
converges in $\bbK$~\cite[p.~13]{Friedlander}.
We have 
\begin{eqnarray*}
\langle v_m,f\rangle =
 \langle T_{x_m}\chi u_{\eta_m,m},f\rangle =
 \langle  u_{\eta_m,m},\chi T_{-x_m}f\rangle =
(2\pi)^{-n} \int_{\bbR^n}
  \widehat{u_{\eta_m,m}}(k)
  \widehat{\chi f_{-x_m}}(-k).
\end{eqnarray*}
where $f_{-x_m}=T_{-x_m}f$.
For every integer $N$ we have by Eq.\eqref{ekNVchi}
\begin{eqnarray*}
|\widehat{\chi f_{-x_m}}(k)| \le (1+|k|)^{-N}(4(n+1)\beta)^N
   |K| \pi_{2N,K}(\chi)
   \pi_{2N,K}(f_{-x_m}),
\end{eqnarray*}
where $K$ is a compact neighborhood of $\supp\chi$
and $|K|$ its volume.

Now, $\pi_{2N,K}(f_{-x_m}) \le \pi_{2N,K'}(f)$,
where $K'$ is a compact neighborhood of $\supp f$.
Thus, there is a constant
$C_N=(4(n+1)\beta)^N |K| \pi_{2N,K}(\chi)\pi_{2N,K'}(f)$,
independent of $m$, such that
$|\widehat{\chi f_{-x_m}}(k)| \le C_N (1+|k|)^{-N}$.
The estimate \eqref{TechIneq2} gives us, for $N=n$,
\begin{eqnarray*}
|\langle v_m,f\rangle| & \le &
C_{n} (2\pi)^{-n} 10^m
\int_{\bbR^n} (1+|k|)^{-n-m} dk
\\&\le &
C_{n} (2\pi)^{-n} 10^m
\int_{\bbR^n} (1+|k|)^{-n-1} dk
\le
C_{n} 10^m I^{n+1}_n,
\end{eqnarray*}
because $m\ge 1$, and the series defining $v$
is absolutely convergent with
$|\langle v,f\rangle| \le C_{n} I^{n+1}_n e^{10}$.

We know that the distribution $v$ is well defined but
we have no control of its wavefront set. Indeed,
the wavefront set of $v$ can contain points that are
not in any $WF(v_m)$ and there can be points that are
in the wavefront set of some $v_m$ but not
in $WF(v)$ (see refs.~\cite{Kashiwara,Hollands2} for 
concrete examples). 
Therefore, we must carefully choose the sequence
$(x_m;\eta_m)$ so that $(x;\eta)$ is indeed in the
wavefront set of $v$. This is done in the next step.

\begin{step}
Choice of the sequence and construction of the cones
\end{step}
We want to ensure that all points $(x_m;\eta_m)$ actually
belong to $\WF(v)$. 
Thus, we choose the elements $(x_m;\eta_m)$ so
that each direction $\eta_m$ is at a finite
distance from the other ones (except when $n= 1$, in which case 
we will choose $x_m$ at a finite
distance from one another), to avoid that their
overlap concurs to remove  $(x;\eta)$ from the
wavefront set of $v$.
Since $\Lambda$ is a cone, we can choose $|\eta|=|\eta_m|=1$
and, up to extraction and since $\Lambda$ is open, it is
possible to shift the points $(x_m;\eta_m)$
so that if $n=1$, $x_m\neq x$ and $\eta_m=\eta$, $|x_{m+1}-x|< |x_m-x|/2, |x_m-x|<1$   and if  $n\neq 1$ $\eta_m\neq \eta$, $|\eta_{m+1}-\eta|<\min(|\eta_m-\eta|,d(\eta_m,
-\Gamma_{x_m}))/2$, where
$\Gamma_{x_m}=\{\xi\telque (x_m;\xi)\in\Gamma\}$,
and $|\eta_m-\eta|<1$ for all $m$. 
Let $\rho_m=\min\big(|\eta_m-\eta|,d(\eta_m,-\Gamma_{x_m})\big)<1$ if $n\neq 1$ and set $\rho_m=1/3^m$ if $n=1$,
and note that if $n\neq 1$,
$\rho_{m+1}< \rho_m/2$ implies $|\eta_m-\eta_k|>\rho_m/2$ for all $k>m$.
Indeed,
if $|\eta_m-\eta_k|\leq \rho_m/2$ were true, 
$\rho_m/2\geq \rho_m/2^{k-m}> \rho_k\geq|\eta_k-\eta|$ 
would imply that
$\rho_m\leq |\eta_m-\eta|\leq |\eta_m-\eta_k|+|\eta_k-\eta|< \rho_m$,
 yielding a contradiction.
Recall that
$v_m=T_{x_m}(\chi u_{\eta_m,m})$ so that
$v_m\in \calE'_\Lambda$ and $\WF(v_m)=\{x_m\}\times \bbR_+^*\eta_m$.

To control the wavefront set, we define 
partial sums $S_m=\sum_{i=1}^m v_i/i!$, and we 
show that the cotangent directions of the
wavefront set of $v-S_m$ do not
meet $(x_i;\eta_i)$ for $i\le m$.
Thus, we have the finite sum 
$v= (v-S_m) + \sum_{i=1}^m v_i/i!$ and,
since the cotangent directions of the wavefront
set of the terms do not overlap, there can be no cancellation
and all $(x_i;\eta_i)$ belong to the wavefront
set for $i\le m$. Then, we have indeed $(x_m;\eta_m)\subset \WF(v)$
for all $m$ because this procedure can be applied for all
values of $m$.

It remains to show that the wavefront set of
$v-S_m$ belongs to a closed conical set $\Gamma_m$
which does not meet $(x_i;\eta_i)$ for $i\le m$.
We first build these $\Gamma_m$ as follows:
Let $X_m=\cup_{l>m}^\infty\{x_l\}\cup \{x\}\subset \Omega$ and $\gamma_{m,i}=X_{m} \times \big(\bbR_+^* \overline{B(\eta_i,\rho_i/4)}\big)$.
It is clear that if $n\neq 1$, $\gamma_{m,i}\cap\gamma_{m,j}=\emptyset$
because, for $j>i$, we have
$|\eta_i-\eta_j|> \rho_i/2$ and $\rho_j<\rho_i$. 
Thus, $|\eta_i-\eta_j|> (\rho_i+\rho_j)/4$ and 
since this expression is symmetric in $i$ and $j$,
it holds for all $i\not= j$.
This shows that the balls 
$\overline{B(\eta_i,\rho_i/4)}$ and 
$\overline{B(\eta_j,\rho_j/4)}$ do not meet and
the result follows.
The closed cones $\gamma_{m,i}$ are then used to
define $\Gamma_m=\big(\bigcup_{i>m}\gamma_{m,i}\big) \cup
(X_{m}\times\bbR_+^*\eta)$.

To show that the wavefront set of $v-S_m$
belongs to $\Gamma_m$, we prove that the
series $\sum_{i=m+1}^\infty v_i/i!$ converges in
$\calD'_{\Gamma_m}$.

\begin{step}
Estimates on seminorms of $v_m$ in $\calD'_{\Gamma_M}$, $m>M$.
\end{step}

Fix $\psi\in \calD(\Omega)$ and any closed  
cone $W$ such that $\supp\psi\times W\cap \Gamma_M=\emptyset$.
For convenience we define the 
distance $||x-y||_\infty=\sup_{i=1,...,n} |x^i-y^i|$,
where $x^i$ is the $i$th coordinate of $x$ in a given
orthonormal basis.
Then, we define the distance between two sets to be
$d_\infty(A,B)=\inf_{x\in A,y\in B} ||x-y||_\infty$.

We first consider the case when
$X_M\cap\supp\psi=\emptyset$. Then, $v_m\psi$ is smooth, and we want 
to show that $\{v_m\psi,m\in \bbN\}$ is bounded in $\calD(\Omega)$, 
since $W$ above can be taken arbitrary.
This is equivalent to prove that $\{\chi\psi_{-x_m}u_{\eta_m,m},m\in \bbN\}$ 
is bounded, where $\psi_{-x_m}=T_{-x_m}\psi$.
Let $\epsilon=d_\infty(X_M,\supp\psi)>0.$ 
Since $\psi$ vanishes in a neighborhood of $x_m$ on the ball 
$B_{\infty}(x_m,\epsilon)$ with $\epsilon>0,$ we deduce that 
$\chi\psi_{-x_m}(y)$ vanishes when $||y||_\infty\leq\epsilon.$
Thus, we can consider that $||y||_\infty / \epsilon \ge 1$.

Then, using the properties of H\"ormander's construction, 
we bound uniformly in $m$. Fix $y$ and choose $y^i$ such that
$|y^i|=||y||_\infty$. Then,
\begin{eqnarray*}
|\partial^\alpha \chi\psi_{-x_m}u_{\eta_m,m}(y)| &\le &
\frac{1}{\epsilon^{C(0,|\alpha|)}}\sum_{\beta\le \alpha} \binom{\alpha}{\beta}
|\partial^\beta \chi\psi_{-x_m}|\,|(y^i)^{C(0,|\alpha|)}
\partial^{\alpha-\beta} u_{\eta_m,m}|
\\&\le& 
\frac{1}{\epsilon^{C(0,|\alpha|)}}2^{| \alpha|}
\pi_{|\alpha|,K_m}(\chi\psi_{-x_m})D(0,|\alpha|),
\end{eqnarray*}
where $K_m=\supp(\psi_{-x_m})$.
To establish Eq.~\eqref{ekNVchi} we showed that 
\begin{eqnarray*}
\pi_{|\alpha|,K_m}(\chi\psi_{-x_m})
&\le & 2^{|\alpha|} 
\pi_{|\alpha|,K_m}(\chi)
\pi_{|\alpha|,K_m}(\psi_{-x_m}).
\end{eqnarray*}
But 
$\pi_{|\alpha|,K_m}(\chi)
\le 
\pi_{|\alpha|,\supp\chi}(\chi)$
and
$\pi_{|\alpha|,K_m}(\psi_{-x_m})=
\pi_{|\alpha|,\supp\psi}(\psi)$.
Thus,
\begin{eqnarray*}
|\partial^\alpha \chi\psi_{-x_m}u_{\eta_m,m}(y)| &\le &
\frac{1}{\epsilon^{C(0,|\alpha|)}}2^{2|\alpha|}
\pi_{|\alpha|,\supp\chi}(\chi)\pi_{|\alpha|,\supp\psi}(\psi)
D(0,|\alpha|)
\end{eqnarray*}
is bounded independently of $m$.

In the case $X_M\cap \supp \psi\neq\emptyset$, we have 
$y\in \supp\psi$ for some $y\in X_M$ and 
$\{y\}\times W \cap \gamma_{M,m}=\emptyset$ for all
$m>M$ by our assumption.
Thus, $W\cap \bbR_+^*\overline{B(\eta_m,\rho_m/4)}=\emptyset$ 
for all $m>M$. Arguing as usual by 
a compactness argument, one can prove that
there is a constant $1>c>0$ 
(independent of $m$) such that for all $k$ satisfying both 
$k\in [\bbR_+^*B(\eta_m,\rho_m/4)]^c$ and
 $(k-q)\in \bbR_+^*B(\eta_m,\rho_m/8)$, 
we have $|q|\geq c\rho_m|k-q|$. 
We will deduce from this and our previous estimates a bound on:

\begin{eqnarray*}
&||v_m||_{N,W,\psi}\leq\sup_{k\in W}(1+|k|)^N
 \int_{\bbR^n}dq|\widehat{u_{\eta_m,m}}(k-q)\widehat{\chi\psi_{-x_m}}(q)|
=I_1+I_2,
\end{eqnarray*}
where $I_1$ corresponds to the integral
over $\Omega_1=\{q\telque \frac{k-q}{|k-q|}\in B(\eta_m,\rho_m/8)\}$ and
$I_2$ over $\Omega_2=\bbR^n\backslash\Omega_1$.
To estimate $I_1$, we use
$|\widehat{u_{\eta_m,m}}(k-q)|\le 10^m (1+|k-q|)^{-m}$
(see Eq.~\eqref{TechIneq2})
and
$(1+|k|)^N\leq (1+|q|)^N(1+|k-q|)^N$ to obtain
\begin{eqnarray*}
I_1 &\le & 10^m\sup_{k\in W}
 \int_{\Omega_1} 
(1+|k-q|)^{N-m} |\widehat{\chi\psi_{-x_m}}(q)|(1+|q|)^N dq.
\end{eqnarray*}
We bound $|\widehat{\chi\psi_{-x_m}}(q)|$
with $||\chi_{x_m}||_{n+1+N,\bbR^n,\psi}(1+|q|)^{-N-n-1}$.
Then, if $N-m \le 0$ we bound $(1+|k-q|)^{N-m}$ with 1
and we obtain $I_1\le 10^m I^{n+1}_n
||\chi_{x_m}||_{n+1+N,\bbR^n,\psi}$,
and if $N-m > 0$, then we bound
$(1+|k-q|)^{N-m}$ with $(c\rho_m)^{m-N}(1+|q|)^{N-m}$
and we find
$$I_1\le 10^m (c\rho_m)^{m-N} I^{n+1}_n
||\chi_{x_m}||_{n+1+2N,\bbR^n,\psi}.$$

To estimate $I_2$, we start as for $I_1$ except that
we use the first inequality of Eq.~\eqref{TechIneq2},
where we replace $80^{m/2}$ by $10^m$:
\begin{eqnarray*}
I_2 &\le & 10^m\sup_{k\in W}
 \int_{\Omega_2} 
  (1-\chi(\langle k-q,\eta_m\rangle))
  \chi(\frac{|k-q|^2-|\langle k-q,\eta_m\rangle|^{2}}
  {|\langle k-q,\eta_m\rangle|^{2\rho}})
\\&&(1+|k-q|)^{N-m}
  |\widehat{\chi\psi_{-x_m}}(q)|(1+|q|)^N dq.
\end{eqnarray*}
By considering the support of $\chi$, we see that
the integrand is zero except (possibly) if 
(i) $\langle k-q,\eta_m\rangle \geq 1/2 $ and
(ii) $|\langle k-q,
\eta_m\rangle|^{2\rho}\ \geq \ |k-q|^2-|\langle k-q,\eta_m\rangle|^{2}$.
Now we show that the three conditions (i), (ii) and
$q\notin \Omega_1$ imply $q\in B(k,r_m)$, where
$r_m=(\rho_m^2/64-(\rho_m^2/128)^2)^{-1/(2-2\rho)}$.
Indeed, $q\notin\Omega_1$ means that
$|(k-q)/|k-q|-\eta_m|^2=2(|k-q|-\langle k-q,\eta_m\rangle)/
|k-q|\geq (\rho_m/8)^2$, so that $|k-q|(1-\rho_m^2/128)\geq \langle k-q,
\eta_m\rangle$. This implies with (i) and (ii):
$|k-q|^2(1-(1-\rho_m^2/128)^2)\leq|k-q|^{2\rho}(1-\rho_m^2/128)^{2\rho}\leq 
|k-q|^{2\rho}$. Thus, $r_m^{2\rho-2} \le |k-q|^{2\rho-2}$
and the result follows because $\rho <1$.
By using $0\le \chi\le 1$ we find
\begin{eqnarray*}
I_2 &\le & 10^m\sup_{k\in W}
 \int_{B(k,r_m)} 
(1+|k-q|)^{N-m}
  |\widehat{\chi\psi_{-x_m}}(q)|(1+|q|)^N dq.
\end{eqnarray*}
We proceed now as for $I_1$ and obtain
$I_2\le 10^m I^{n+1}_n
||\chi_{x_m}||_{n+1+N,\bbR^n,\psi}$ if $N-m \le 0$
and
$I_1\le 10^m (1+r_m)^{N-m} I^{n+1}_n||\chi_{x_m}||_{n+1+N,\bbR^n,\psi}$,
if $N-m > 0$.
We have showed that
$||\chi_{x_m}||_{n+1+N,\bbR^n,\psi}$ can be bounded independently
of $m$. Thus, for $m >N$, there is a constant $C_{n,N}$
such that
$||v_m||_{N,W,\psi} \le 10^m C_{n,N}$. Since the set
of $m\le N$ is finite, we see that
$10^{-m} ||v_m||_{N,W,\psi}$ is bounded for all values of $m$.

Thus, we showed that, for any $W$ and $\psi$ such that
$\supp\psi\times W \cap \Gamma_M=\emptyset$ and 
any integer $N$, the set
$\{10^{-m} ||v_m||_{N,W,\psi}\telque m>M\}$ is bounded in $\bbR$.
To show that the set 
$A=\{10^{-m} v_m\telque m>M\}$ is bounded in 
$\calD'_{\Gamma_M}$, we still have to show that it
is bounded for the seminorms $p_B$ with $B$
bounded in $\calD(\Omega)$. 
In the course of step~2, we showed 
that, for any $f\in \calD$, the set
$p_f(A)$ is bounded in $\bbR$. This means that
$A$ is bounded in $\calD'_{\Gamma_M}$ equipped
with the H\"ormander topology. But we proved that
this is equivalent to being bounded for the
normal topology. Thus, $A$ is bounded
in $\calD'_{\Gamma_M}$ with its normal topology.

\begin{step}
Let  $S_m:=\sum_{k=1}^m\frac{1}{k!}v_k$ ($S_0=0$). 
Then for any $M\ge 0$, the
sequence $(S_m-S_M)_{m\ge M}$ is a Cauchy sequence in 
$\calD'_{\Gamma_M}.$ As a consequence, 
$S_m-S_M$ converges 
to $v-S_M$ in 
$\calD'_{\Gamma_M}$ and $WF(v)\supset\{(x_m;\eta_m),m\in\bbN^*\}.$ 
\end{step}

In the previous step we showed that the set
$A=\{10^{-m} v_m\telque m>M\}$ is bounded in 
$\calD'_{\Gamma_M}$. Thus, for every
seminorm $p_i$ of $\calD'_{\Gamma_M}$
and any $p\ge q >  M$, we have
$p_i(S_p-S_q) \le C_i \sum_{m=q}^p 10^m/m!$, and
each $p_i(S_m-S_M)$ is a Cauchy sequence in $\bbR$.
By the completeness of $\calD'_{\Gamma_M}$,
it implies that $S_m-S_M$ converges
to $v-S_M$ in $\calD'_{\Gamma_M}$.

Since the wavefront set is known for each $v_m$
($WF(v_m)=\{x_m\}\times \bbR_+^*\eta_m$),  $v_M$ is the only 
one among the distributions $v-S_M,v_M,...,v_1$ which is
singular in direction $\bbR_+^*\eta_M$ at $x_M$ (because 
$(x_M;\eta_M)\notin\Gamma_M$ by construction either because $x_m\neq x_M$ if $n=1$ or because $\eta_m\neq \eta_M$ if $n\neq 1$, for $m>M$),
one deduces $\{x_M\}\times \bbR_+^*\eta_M\subset WF(v)$.
Indeed, by choosing a test function $\psi$ such that
$\psi(x_M)\not=0$ and a closed cone $V\subset \bbR_+^* B(\eta_M,\rho_M/4)$,
we have $\supp\psi\times V\cap \WF(v_m)=\emptyset$ for $m<M$
and $\supp\psi\times V\cap \WF(v-S_M)=\emptyset$.
Therefore, $||v-v_M||_{N,V,\psi}$ is finite for all $N$ and
$\widehat{\psi(v-v_M)}(\lambda \eta_M)$ cannot compensate for 
the slow decrease of $\widehat{\psi v_M}(\lambda \eta_M)$,
which is ensured by the fact that
$\WF(\psi v_M)=\WF(v_M)$ when
$\psi(x_M)\not=0$~\cite[p.~121]{Hormander-71}.
Since 
this is valid for any $M$, this proves the wavefront set statement.

It remains to show that the sequence is also Cauchy in $\calE'_\Lambda$.

 \begin{step}
$S_m:=\sum_{k=1}^m\frac{1}{k!}v_k$  is Cauchy in $\mathcal{E}'_{\Lambda}$ 
for the strong topology coming from its duality with $\calD'_\Gamma$ 
(and even Mackey-Cauchy for the corresponding von Neumann bornology). 
Especially, $\mathcal{E}'_{\Lambda}$ is not sequentially complete 
(and not even Mackey-complete).
\end{step}
By construction $WF(S_m)\subset \Lambda$. Assume proved the statement 
about its Cauchy nature, then the last step enables to show that if it were 
(even weakly) convergent in $\mathcal{E}'_{\Lambda}$, then the limit would be 
$v$ (since it would be weakly convergent in 
$\mathcal{D}'_{\overline{\Lambda}}$ where the limit is $v$) 
as a distribution, but since the wavefront set is closed, 
$(x;\eta)\in WF(v)$ and since $(x;\eta)\not\in \Lambda$ 
this gives a contradiction, implying that $S_m$ is a Cauchy 
sequence not (weakly) converging in $\mathcal{E}'_{\Lambda}$.

Thus it remains to show that $S_m$ is Cauchy. 
Take $B\subset \calD'_\Gamma$ 
bounded, we want to show that 
$p_B(S_m)=\sum_{k=1}^m
p_B(v_k)/k!$ is a 
Cauchy sequence.
First choose $\tilde{\chi}\in D(\Omega)$ which is identically one on the 
compact set $\cup_{y\in X} \supp\chi_y=X+\supp\chi$ (the sum of 
two  compact 
sets is compact), where $\chi_y=T_y\chi$.  
Using lemma~\ref{UniformBound}, since $B$ is bounded in $\calD'(\Omega)$, 
fix $M=M_B$ such that $\sup_{u\in B}\sup_{k\in\bbR^n}
(1+|k|)^{-M}|
\widehat{\tilde{\chi}u}(k)|=D<\infty.$ 
Then, for $y\in X$, we bound: 
\begin{align*}
\sup_{u\in B}||u||_{M,\bbR^n,\chi_y} &=
\sup_{u\in B}\sup_{k\in\bbR^n}(1+|k|)^{-M}
|\widehat{\chi_yu}(k)|=
\sup_{u\in B}\sup_{k\in\bbR^n}(1+|k|)^{-M}
|\widehat{\chi_y\tilde{\chi}u}(k)|\\&
\leq \sup_{u\in B}\sup_{k\in\bbR^n}\int_{\bbR^n} dq 
(1+|k-q|)^{M}|\widehat{\chi_y}(k-q)
\widehat{\tilde{\chi}u}(q)|(1+|q|)^{-M}
\\& \leq DI_{n}^{n+1}((1+n)\beta)^{(M+n+1)}\pi_{2(M+n+1),\supp(\chi)}(\chi)
=C_B<\infty.\end{align*}

It now suffices to estimate $p_B(v_m)$ for $m\geq M+n+1.$
Thus, using this inequality and (\ref{TechIneq2}), we deduce for $m\geq M+n+1$:
 
\begin{align*}\sup_{u\in B}|\langle u,v_m\rangle|
&=\sup_{u\in B}\frac{1}{(2\pi)^n}\left|\int_{\bbR^n}dk 
  \widehat{\chi T_{-x_m}u}(k)
  \widehat{u_{\eta_m,m}}(-k)\right|\\&\leq C_B10^m
 \int_{\bbR^n}dk(1+|k|)^M(1+|k|)^{-m}\leq C_B10^mI_{n}^{n+1}.
 \end{align*}
Thus, for $p\geq q\geq M+n+1$, $p_B(S_p-S_q)\leq C_BI_{n}^{n+1}
\sum_{k=q+1}^p\frac {10^k}{k!},$ and $p_B(S_m)$ is Cauchy as we wanted.

More precisely, let us define the following bounded set for the strong topology of $\mathcal{E}'_{\Lambda}$ $$A'=\{v\in \mathcal{E}'_{\Lambda} : p_B(v)\leq \max(C_BI_{n}^{n+1},\max_{m\leq M_B+n+1}(p_B(v_m)))\ \forall B\ \text{bounded in }\ \calD'_\Gamma(\Omega) \}.$$
Note that if $q\leq M_B+n+1\leq p$, or $q\leq p\leq M_B+n+1$, we still have  $$p_B(S_p-S_q)\leq  
\sum_{k=q+1}^p\frac{10^k}{k!}\max(C_BI_{n}^{n+1},\max_{m\leq M_B+n+1}(p_B(v_m))).$$
Thus, if $\lambda_{p,q}=\lambda_{q,p}=
\sum_{k=q+1}^p\frac{10^m}{m!}$ for $p\geq q$, we showed 
$S_p-S_q\in \lambda_{p,q} A'$ and since
$\lambda_{p,q}\underset{p,q\to\infty}{\rightarrow} 0,$ we even
deduce that $S_p$ is Mackey-Cauchy. This concludes.

\begin{step}
Characterization of closed $\Lambda$.
\end{step}
To complete the proof we give some information
on the case when $\Lambda$ is open and closed.
A subset of a topological space $X$ is called \emph{clopen} 
if it is both open and closed in $X$~\cite[p.~10]{Bredon}.
A topological space $X$ is connected if and only if its only clopen
subsets are $X$ and $\emptyset$~\cite[p.~10]{Bredon}.
Now, if $\Omega$ is connected, its cotangent bundle
$T^*\Omega$ is connected. If the dimension of
$\Omega$ is $n>1$ the set 
$\dotT^*\Omega$, which is $T^*\Omega$ with the
zero section removed, is also connected.
In that case $\Lambda$ is clopen if and only if it is either
empty (so that $\calE'_\Lambda=\calD(\Omega)$)
or $\dotT^*\Omega$
(so that $\calE'_\Lambda=\calE'(\Omega)$).
Since both $\calD(\Omega)$ and $\calE'(\Omega)$ are complete,
our theorem is optimal for connected $\dotT^*\Omega$.
\end{proof}
\begin{cor}
If $\Lambda$ is an open cone which is not closed, $\calE'_\Lambda$ is not sequentially complete for any topology
that is coarser than the normal topology and finer than the
weak topology of distributions induced by $\calD'(\Omega)$. In particular, the inductive
limit of $E_\ell$ equipped with the H\"ormander topology
is also not sequentially complete.
\end{cor}
\begin{proof}
This result is a consequence of the proof above rather than of the statement.
A sequence which is Cauchy for the normal topology 
remains Cauchy for topologies that are coarser than
it, thus our counterexample above is  Cauchy for the topologies considered. 
Therefore, it converges weakly in $\calD'(\Omega)$ and we showed 
that the limit cannot be in $\calE'_\Lambda$ so that 
$\calE'_\Lambda$ is not sequentially complete. 
\end{proof}

\begin{cor}
If $\Lambda$ is an open cone which is not closed, 
then $\calE'_\Lambda$ 
is not a regular inductive limit for the inductive topology
(which is equivalent to the strong topology) and it is not semi-reflexive. 
If $(\Gamma')^c=\Lambda$, $\calD'_{\Gamma}$ is neither bornological 
nor barrelled in its normal topology
\end{cor}

\begin{proof}
If $\calE'_\Lambda$ were semi-reflexive it would be weakly 
sequentially complete~\cite[p.~228]{Horvath}. If the inductive 
limit were regular, it would be semi-reflexive as explained at 
the end of section~4.3. Alternatively, one can see that the set of the Cauchy sequence $\{S_m,m\geq 1\}$ we built is bounded in $\calE'_\Lambda$ and not in any $E_\ell$.

The space $\calD'_\Gamma$ is not bornological because 
the strong dual of a separated bornological space is 
complete~\cite[p.~77]{Hogbe-77}.
If $\calD'_{\Gamma}$ were barrelled in its normal topology so that, 
since it is semi-Montel, it would be a Montel space~\cite[p.~231]{Horvath}, 
then its strong dual $\calE'_\Lambda$ would also be a 
Montel space~\cite[p.~234]{Horvath} and thus again semi-reflexive.
Note that Bourbaki states that a space that is semi-reflexive 
and semi-barrelled is complete~\cite[p.~IV.60]{Bourbaki-TVS},
but this is wrong~\cite{Bourles-13}.
\end{proof}

Remark that $\calD'_\Gamma$ provides a concrete and natural 
example of a complete nuclear space whose strong dual is not
sequentially-complete.
Grothendieck constructed other examples by using
sophisticated techniques of topological tensor
products~\cite[Ch.~II, p.~83 and p.~92]{Grothendieck-55}
(see also \cite[p.~96]{Hogbe-73-nuclear}).

\section{Conclusion}

This paper determined the main functional properties
of H{\"o}rmander's space of distributions $\calD'_\Gamma$
and its dual.
In view of applications to the causal approach of
quantum field theory, we derived 
simple rules to determine whether
a distribution belongs to $\calD'_\Gamma$, whether
a sequence converges in $\calD'_\Gamma$ 
and whether a subset of $\calD'_\Gamma$ is bounded.
The properties of $\calD'_\Gamma$ can also
be useful to other physical applications where the wavefront set
played a crucial 
role~\cite{Ivrii-77,Esser-87,Kay-97,Fewster-00,Strohmaier-02,%
Franco-07,Pinamonti-11}.

By using the functional properties of $\calD'_\Gamma$, the proof of
renormalizability of scalar quantum field theory in
curved spacetime can be considerably simplified and
streamlined with respect to the original derivation
given by Brunetti and Fredenhagen~\cite{Brunetti2}.

The results of the present paper will be
extended in two directions: i) 
The continuity properties of the main operations with
distributions in $\calD'_\Gamma$
(tensor product, pull-back, push-forward,
multiplication of distributions)~\cite{Viet-wf2};
ii) A detailed investigation of the microcausal functionals discussed by
Brunetti, D{\"u}tsch, Fredenhagen, 
Rejzner and Ribeiro~\cite{Brunetti-09,Fredenhagen-11,%
Brunetti-12,Brunetti-13-QG}, which are the basis
of a new and powerful formulation of quantum field theory.
As noticed in ref.~\cite{Brunetti-12}, the space of
microcausal functionals is based on spaces 
of the type $\calE'_\Lambda$ which have very
poor completeness properties. 
This problem can be solved by using the completion
of $\calE'_\Lambda$, which is, because of the nuclearity of $\calE'_\Lambda$, also 
the bornological dual of $\calD'_\Gamma$~\cite[p.~140]{Hogbe-81}.
The topological and bornological properties of this completion will
be discussed in a forthcoming publication 
by the first author~\cite{Dabrowski-13}.

\section{Acknowledgements}
The authors were partially supported by GDR Renormalisation.
We thank Christian G\'erard for useful remarks.
The first author wants to thank the organizing committee of the 2012 Les Houches Winter School on ``New mathematical aspects of quantum field theories" supported by the LABEX MILYON (ANR-10-LABX-0070), where he started learning about quantum field theory in curved spacetime.
The second author is grateful to Katarzyna Rejzner,
Nguyen Viet Dang, Thierry De Pauw, 
Fr\'ed\'eric H\'elein, Christopher Boyd,
Michael Oberguggenberger and Alain Grigis for discussions
at an early stage of this work.


\end{document}